\documentclass[reprint]{revtex4-1}
\usepackage{amsmath,amssymb,amsthm}
\usepackage{bm}
\usepackage[shortlabels]{enumitem}
\usepackage{tabularx}
\usepackage{graphicx}
\usepackage{xcolor}

\DeclareMathOperator{\cdf}{cdf}
\DeclareMathOperator{\tr}{tr}
\DeclareMathOperator{\proj}{proj}
\DeclareMathOperator{\spn}{span}
\newtheorem{thm}{Theorem}

\begin{document}
\title{Quantum mechanics and data assimilation}
\author{Dimitrios Giannakis}
\date{\today}
\affiliation{Center for Atmosphere Ocean Science, Courant Institute of Mathematical Sciences, New York University, New York, New York 10012, USA}

\begin{abstract}
    A framework for data assimilation combining aspects of operator-theoretic ergodic theory and quantum mechanics is developed. This framework adapts the Dirac--von Neumann formalism of quantum dynamics and measurement to perform sequential data assimilation (filtering) of a partially observed, measure-preserving dynamical system, using the Koopman operator on the $L^2$ space associated with the invariant measure as an analog of the Heisenberg evolution operator in quantum mechanics. In addition, the state of the data assimilation system is represented by a trace-class operator analogous to the density operator in quantum mechanics, and the assimilated observables by self-adjoint multiplication operators. An averaging approach is also introduced, rendering the spectrum of the assimilated observables discrete, and thus amenable to numerical approximation. We present a data-driven formulation of the quantum mechanical data assimilation approach, utilizing kernel methods from machine learning and delay-coordinate maps of dynamical systems to represent the evolution and measurement operators via matrices in a data-driven basis. The data-driven formulation is structurally similar to its infinite-dimensional counterpart, and shown to converge in a limit of large data under mild assumptions. Applications to periodic oscillators and the Lorenz 63 system demonstrate that the framework is able to naturally handle highly non-Gaussian statistics, complex state space geometries, and chaotic dynamics. 
\end{abstract}

\maketitle

\section{Introduction}
Data assimilation is a framework for state estimation and prediction for partially observed dynamical systems \cite*{MajdaHarlim12,LawEtAl15}. Its sequential formulation, also known as filtering, is based on a predictor-corrector procedure, whereby a forward model is employed to evolve the probability distribution for the system state until a new observation is acquired, at which time that probability distribution is updated in an analysis step to a posterior distribution correcting for model error and/or uncertainty in the prior distribution. Since the seminal work of Kalman \cite{Kalman60} on filtering (which utilizes Bayes' theorem for the analysis step, under the assumption that all distributions are Gaussian), data assimilation has evolved to an indispensable tool in virtually every modeling scenario for complex systems, including object tracking \cite{ThrunEtAl05}, weather forecasting \cite{Kalnay03}, and many other important applications \cite{LahozEtAl10}. 

In certain aspects, the predictor-corrector approach in data assimilation resembles another extremely successful branch of modern science, namely, quantum mechanics. Between measurements, the quantum mechanical state evolves under unitary dynamics through the Heisenberg operators, while the measurement process is described by projective dynamics (the so-called wavefunction collapse). As is well known, a fundamental difference between quantum and classical physics is that the quantum mechanical observables are represented by linear operators on a Hilbert space, as opposed to functions on state space in classical physics. In particular, quantum mechanical observables may be non-commuting, and this fundamentally affects the evolution of uncertainty in a quantum system. 

Yet, despite these differences with classical physics, the unitary and projective dynamics underpinning quantum mechanical systems bear some conceptual similarity with the forecast and analysis steps in filtering, respectively, even if the underlying dynamical system is deterministic (i.e., ``classical''). The goal of this work is to explore whether these conceptual similarities can be extended to the level of a mathematically precise data assimilation framework. In fact, we will formulate such a framework by literally transcribing the axioms of quantum mechanics to a partially observed dynamical system as in the setting of data assimilation. 

This framework, which we refer to as quantum mechanical data assimilation (QMDA), can naturally handle a number of challenges encountered by classical data assimilation schemes. In particular, in many real-world applications, rigorous Bayesian approaches (implemented, e.g., via particle filters \cite{VanLeeuwenEtAl19}) become intractable, and as a result ad hoc approximation schemes are commonly employed in both the forecast and analysis steps \cite{LawStuart12}. These schemes oftentimes impose various types of Gaussianity assumptions,  with difficult to to control convergence properties, particularly in the presence of complex deterministic dynamics exhibiting features such as fractal attractors and singular probability measures. On the other hand, QMDA employs finite-rank approximations of the intrinsic evolution and measurement operators of such systems, realized through Koopman operator theory \cite{BudisicEtAl12,EisnerEtAl15} and kernel methods for machine learning \cite{BelkinNiyogi03,CoifmanLafon06,VonLuxburgEtAl08,BerryHarlim16},  with well-established convergence properties. It should be noted that while connections between quantum theory and data assimilation have been studied in the literature \cite{Accardi91,EmzirEtAl17}, these works have generally approached the problem of performing data assimilation for an actual physical quantum system. To our knowledge, the approach presented here, which combines the Koopman operator formalism with abstract quantum mechanical axioms to construct a data assimilation algorithm for deterministic dynamical systems, as well as its approximation via machine-learning techniques, has not been studied elsewhere. 

The plan of this paper is as follows. In Section~\ref{secQMDA}, we describe the basic mathematical formulation of the QMDA approach. In Section~\ref{secCircle}, we illustrate the behavior of this framework in a simple example involving a periodic dynamical system observed through a binary observation function. In Section~\ref{secCompactification}, we consider data assimilation of observables with potentially continuous spectrum, and present an averaging approach to render their spectra discrete. In Section~\ref{secDataDriven}, we describe the data-driven formulation of our schemes using kernel algorithms. The data-driven approach is demonstrated in Section~\ref{secL63} in the context of the partially observed Lorenz 63 (L63) system. Section~\ref{secDiscussion} discusses some aspects of QMDA in relation to classical data assimilation methodologies. Our primary conclusions are stated in Section~\ref{secConclusions}. A technical result on convergence of the data-driven formulation of QMDA is stated and proved in Appendix~\ref{appConvergence}. Appendix~\ref{appComputational} contains a discussion on numerical implementation and computational cost, along with formulas for the QMDA steps expressed in matrix algebra.  

\section{\label{secQMDA}Quantum mechanical formulation of data assimilation}

We begin by reviewing the axioms of quantum mechanics according to the canonical Dirac--von Neumann formulation \cite{Takhtajan08}. 
\begin{enumerate}[{QM}1), wide]
\item Associated with every quantum system is a separable Hilbert space $( H, \langle \cdot, \cdot \rangle_H)$ over the complex numbers. The possible states of the system correspond to the set of non-negative, trace-class operators $ \rho : H \to H $, such that $ \tr \rho = 1 $. The observables of the system are self-adjoint linear operators on $ H$. We will denote the sets of bounded and trace-class operators on a Hilbert space $H$ by $ B(H) $ and $ B_1(H)$, respectively. 
\item Between measurements, the state evolves under the action of a strongly continuous group of unitary operators $U^t : H \to H $, $t \in \mathbb{R}$, called Heisenberg operators. Specifically, the state $ \rho_t $ reached at time $ t $ starting from a state $ \rho_0 \in B_1(H) $ is given by 
\begin{displaymath}
\rho_t = U^{t*} \rho_0 U^t. 
\end{displaymath}
\item Let $ A : D(A) \to H $ be an observable, defined on a dense subspace $D(A) \subseteq H $. By the spectral theorem for self-adjoint operators, there exists a unique projection-valued measure $ E_A : \mathcal{B}(\mathbb{R}) \to B(H) $ on the Borel $ \sigma $-algebra $ \mathcal{B}(\mathbb{R}) $ on $ \mathbb{R}$, such that $ A = \int_\mathbb{R} a \, dE_A(a)$. The set of possible values that a measurement of $A $ can take in a physical experiment is given by the spectrum of $ A$, $ \sigma(A) \subseteq \mathbb{R}$. 
\item If the system is in state $ \rho \in B_1(H)$, then the probability that a measurement of an observable $ A $ will yield a value lying in a Borel set $\Omega \subseteq \mathbb{R}$ is equal to  $ \tr( E_A( \Omega ) \rho ) $.
\item If the system state immediately before a measurement is $ \rho^-$, and a measurement of $ A $ yields a value $ a \in \sigma(A)$, with $E_A(\{a\}) \neq 0 $ (i.e., $a$ is an eigenvalue of $A$), then the state $ \rho^+ $ immediately after the measurement is given by
\begin{displaymath}
\rho^+ = \frac{ E_A(\{a\}) \rho^- E_A(\{a\}) }{ \tr (E_A(\{a\}) \rho^- E_A(\{a\}) ) }.
\end{displaymath}    
\end{enumerate}

Axioms QM2 and QM5 describe the unitary and projective parts of quantum dynamics, respectively. Note that we have stated QM5 only in the case of measurements lying in the point spectrum of $A$. This will be sufficient for our purposes, since the QMDA framework will employ an averaging procedure, approximating the measurement operator in data assimilation by a self-adjoint operator with pure point spectrum.

We now consider how to construct a data assimilation scheme that mimics the quantum mechanical axioms listed above. In this construction, we will assume that the dynamics is described through a continuous measure-preserving flow $ \Phi^ t : M \to M $, $ t\in \mathbb{R} $, on a metric space $M$, with an ergodic, invariant, compactly supported  Borel probability measure $\mu$. Associated with the flow $\Phi^t$ is a unitary group of Koopman evolution operators \cite{Koopman31,BudisicEtAl12,EisnerEtAl15}, acting on vectors in  $L^2(\mu) $ by composition, $U^t f = f \circ \Phi^t$. We consider that the system is observed through a real-valued, bounded measurement function $ h \in L^\infty(\mu) $. With these definitions, the data assimilation analogs of the quantum mechanical axioms above are as follows.
\begin{enumerate}[{DA}1), wide]
\item Associated with the data assimilation system is the separable Hilbert space $L^2(\mu)$, equipped with the standard inner product, $ \langle f, g \rangle_\mu = \int_M f^* g \, d\mu$. The state of the system lies in the set of non-negative, trace-class operators $ \rho \in B_1(L^2(\mu))$, such that $ \tr \rho = 1 $. The observables of the data assimilation system are self-adjoint linear operators on $L^2(\mu)$. In particular, associated with the measurement function $ h $ is a self-adjoint multiplication operator $ T_h \in B(L^2(\mu) ) $, such that  
        \begin{displaymath}
            T_h f = h f. 
        \end{displaymath}
        
    \item Between measurements, the state evolves under the action of the unitary Koopman operators $ U^t : L^2(\mu) \to L^2(\mu) $ induced by the dynamical flow. In particular, the state reached at time $ t $ starting from a state $ \rho_0 \in B_1(L^2(\mu)) $ is given by 
\begin{displaymath}
\rho_t = U^{t*} \rho_0 U^t.
\end{displaymath}
\item Let $ A : D(A) \to L^2(\mu) $ be an observable with the corresponding projection-valued measure $ E_A : \mathcal{B}(\mathbb{R}) \to B(L^2(\mu)) $. The set of values of $A$ that can be observed with nonzero probability is given by the spectrum $ \sigma(A) $. In particular, in the case of the multiplication operator $T_h $, the spectrum $ \sigma(T_h)$ coincides with the essential range of $h$.  We will use the notation $ E_{h} \equiv E_{T_h} $ to represent the projection-valued measure associated with a real multiplication operator $T_h$.  
    
\item If the data assimilation system has state $ \rho \in B_1(L^2(\mu))$, then the probability that a measurement of  $A $ will yield a value lying in a Borel set $\Omega \subseteq \mathbb{R} $ is equal to $ \tr( E_A(\Omega) \rho ) $. 
\item If the data assimilation state immediately before a measurement is $ \rho^- \in B_1(L^2(\mu)) $, and a measurement of $ A $ yields the value $ a \in \sigma(A)$, with $E_A(\{a\}) \neq 0 $, then the state $ \rho^+ $ immediately after the measurement is given by
\begin{displaymath}
\rho^+ = \frac{ E_A(\{a\}) \rho^- E_A(\{a\}) }{ \tr (E_A(\{a\}) \rho^- E_A(\{a\}) ) }.
\end{displaymath}
\end{enumerate}

A comparison between the ``classical'' and ``quantum'' formulations of sequential data assimilation is displayed in Table~\ref{tableComparison}. There, it can be seen that QMDA reformulates the forward dynamics and Bayesian analysis steps in classical data assimilation using the Koopman operator $U^t$ and the spectral projectors $ E_h(\{a\}) $, respectively, both of which are intrinsically linear. We now discuss some of the general properties of this scheme, which we will expand upon and demonstrate with numerical experiments in the ensuing sections.

\begin{table*}
    \caption{\label{tableComparison}Comparison between the ``classical'' and ``quantum'' formulations of sequential data assimilation for a bounded measurement function $ h \in L^\infty(\mu) $. In the classical formulation, $ \mathcal{P}(M)$ denotes the set of Borel probability measures on $M$. Moreover, $ \nu_a : \mathcal{B}(M) \to \mathbb{R}$, $ a \in \mathbb{R}$, denotes the Borel measure on $M$ satisfying $ \nu_a(\Omega) = \nu( \Omega \cap 1_{h^{-1}(\{a\})} ) $. Note that the projective dynamics step in the classical formulation is the Bayesian update rule. In both the classical and quantum formulations, the projective dynamics steps are only well-defined if the denominators in the respective formulas are non-vanishing.}
    \begin{tabular*}{\linewidth}{@{\extracolsep{\fill}}lll}
\hline
& Classical & Quantum \\
\hline
State & Probability\ measure $ \nu \in \mathcal{P}(M) $ & Trace-class operator $ \rho \in B_1( L^2(\mu) ) $ \\
Observable & $ h \in L^\infty(\mu) $ & $ T_h \in B( L^2(\mu ) ) $ \\
Evolutionary dynamics & $ \nu \mapsto \nu \circ \Phi^{-t} $ & $ \rho \mapsto U^{t*} \rho U^{t} $ \\
Measurement probability & $ \nu( 1_{h^{-1}(\Omega)} ) $ & $ \tr( E_{h}(\Omega) \rho ) $ \\
Projective dynamics & $ \nu \mapsto \frac{ \nu_a }{ \nu_a(M) } $ & $ \rho \mapsto \frac{ E_{h}(\{a\}) \rho E_{h}(\{a\}) }{ \tr (E_{h}(\{a\}) \rho E_{h}(\{a\}) ) } $\\
\hline
\end{tabular*}
\end{table*}

First, it should be noted that, as in our statement of the quantum mechanical axioms, we have stated the state update in DA5 only for measurements lying in the point spectrum, denoted $ \sigma_p(A) $, of the observable $A$. When $A $ is a multiplication operator $T_h$ associated with a bounded measurement function $h$, as would be the case in typical data assimilation scenarios, the condition that  $ a \in \sigma_p(A) $  is equivalent to the subset $ h^{-1}(\{a\}) \subseteq M $ of state space  having positive $ \mu$-measure. 

An observable $A$ is said to have pure point spectrum if there exists an orthonormal basis of $L^2(\mu)$ consisting of its eigenfunctions. In that case, $ \sigma_p(A) $ is a dense subset of $\sigma(A) $, so that every measurement of $A $ is arbitrarily close to an eigenvalue. Examples of measurement functions $h$ resulting in $ A=T_h$ with pure point spectrum are indicator functions of non-null subsets of $M$, representing binary measurements. Indicator functions are in turn special cases of simple (``quantized'') functions taking finitely many values,  where $T_h$ has again pure point spectrum. Such functions are appropriate for modeling experimental scenarios with detectors of finite resolution and dynamic range. In contrast, if there exists $a \in \sigma(T_h)$ such that $ \mu( h^{-1}(\{a\}) ) $ vanishes, then $E_h(\{a\}) $ also vanishes and $ a $ lies in the continuous spectrum of $T_h$. Clearly, as with quantum mechanical axiom QM5, for such measurements the update formula in DA5 is not applicable. We will discuss how to address this situation in Section~\ref{secCompactification} below. For now, observe that for an arbitrary self-adjoint multiplication operator $  T_h$, the spectral projection $ E_h(\Omega) $ associated with a Borel subset $ \Omega \subseteq \mathbb{R} $ is itself a multiplication operator; specifically,
\begin{displaymath}
    E_h(\Omega) = T_{1_{\Omega \cap h(M)}},
\end{displaymath}
where $1_{S} : M \to \mathbb{R}$ denotes the characteristic function of any set $S \subseteq M $. It follows from the above that $E_h(\Omega) $ vanishes whenever $ \mu(h^{-1}(\Omega)) = 0 $, which includes the case discussed above with $ \Omega = \{ a \} \subseteq \sigma(T_h) $ and $a$ lying in the continuous spectrum of $T_h$.

Next, observe that because every data assimilation state $ \rho \in B_1(L^2(\mu))$ is a non-negative operator with unit trace, its diagonal elements $ \varrho_j = \langle \phi_j, \rho \phi_j \rangle_\mu $ in any orthonormal basis $ \{ \phi_j \}_{j=0}^\infty $ of $L^2(\mu)$ correspond to the density of a probability measure $ \varrho $ on the non-negative integers, $ \mathbb{N}_0 $ (i.e., the indexing set of the basis); in particular, we have $ \varrho_j \geq 0 $ and $ \sum_{j=0}^\infty \varrho_j = 1 $.  Adopting quantum mechanical terminology, we will say that  $ \rho $ is a pure state if there exists $ f \in L^2(\mu) $ such that $ \rho = \langle f, \cdot \rangle_\mu f $, and will otherwise refer to it as mixed. If $ \rho = \langle f, \cdot \rangle_\mu f $ is pure, then in any orthonormal basis of $ L^2(\mu) $ having $ f $ as one of its elements $ \varrho$ becomes a Dirac $ \delta$-measure.    In step DA5, if $a $ is a simple eigenvalue of $A $, then the state $ \rho^+ $ following a measurement of $A$ yielding the value $a $  will be pure; otherwise, $ \rho^+ $ will be generally mixed. On the other hand, the unitary evolution between measurements in DA2 always maps pure states to pure states.  

Note now that it is a standard result from ergodic theory \cite{EisnerEtAl15} that, that the Koopman group $ \{ U^t \}_{t\in\mathbb{R}} $ has a simple eigenvalue equal to 1, with a constant corresponding eigenfunction equal to $1_M$. It is straightforward to verify that the corresponding pure state, $ \bar \rho = \langle 1_M, \cdot \rangle_\mu 1_M $, satisfies
\begin{displaymath}
    \tr( E_h( \Omega ) \bar \rho ) = \int_\Omega d\mu_h  
\end{displaymath}
for every measurement function $ h \in L^\infty(\mu) $ and Borel set $ \Omega \subseteq \mathbb{R} $, where $ \mu_h : \mathcal{B}(\mathbb{R}) \to [ 0, 1 ] $ is the pushforward probability measure induced on the real line by $h$ and the invariant measure, satisfying $ \mu_h(\Omega) = \mu(h^{-1}(\Omega)) $. As a result, all probabilities computed via step DA4 for the state $ \bar \rho $ (and thus all statistics such as expectation values, variances, etc., derived from it) are equivalent to probabilities/statistics computed with respect to the stationary distribution of $h$ viewed as a random variable (i.e., $ \mu_h$). For this reason, we refer to $ \bar \rho $ as the stationary state of the data assimilation system. 

As a final general remark, it is worthwhile noting that even though the focus of this work is largely on observables associated with multiplication operators $T_h$, which have an underlying ``classical'' observable $h$, our framework is also applicable to general observables $A $ with no classical counterparts. In the context of measure-preserving, ergodic dynamical systems with strongly continuous unitary Koopman groups, a natural such observable is the \emph{generator} of the Koopman group. In particular, it follows by Stone's theorem for one-parameter unitary groups \cite{Stone32} that there exists a skew-adjoint operator $V:D(V) \to L^2(\mu)$, defined on a dense domain $D(V) \subset L^2(\mu)$ via
\begin{displaymath}
    Vf  = \lim_{t \to 0} \frac{U^t f - f}{ t}, \quad \forall f \in D(V).
\end{displaymath}
This operator generates the Koopman group, in the sense that $ U^t = e^{tV} $, with operator exponentiation computed through the spectral theorem for skew-adjoint operators. In particular, after multiplication by the imaginary number $i $ to render it self-adjoint, $V$ behaves analogously to the Hamiltonian operator in quantum physics, which generates the unitary group of Heisenberg operators. In light of this analogy, $ V $ can be viewed as an energy observable for the data assimilation system, which has no classical counterpart associated with a multiplication operator.    

\section{\label{secCircle}Demonstration in a simple ergodic dynamical system}

In this section, we demonstrate the framework described in Section~\ref{secQMDA} in the context of a simple measure-preserving, ergodic dynamical system, namely, a rotation on the circle, $M=S^1$. In this case, the dynamical flow $\Phi^t : M \to M $ is given by
\begin{displaymath}
\Phi^t( \theta ) = \theta + \omega t \mod 2 \pi,
\end{displaymath}
where $ \omega \in \mathbb{R}$ is a frequency parameter. This system has a unique ergodic invariant Borel probability measure $ \mu $, equal to the Haar measure on $S^1$. The corresponding Koopman operators $U^t : L^2(\mu) \to L^2(\mu) $ have a pure point spectrum, with an associated orthonormal basis of  $L^2(\mu)$, $ \{ \ldots, \phi_{-1},\phi_0,\phi_1, \ldots \} $, consisting of Koopman eigenfunctions,
\begin{displaymath}
\phi_j(\theta) = e^{i j \theta}, \quad U^t \phi_j = e^{ij\omega t} \phi_j, \quad \langle \phi_j, \phi_k \rangle_\mu = \delta_{jk}. 
\end{displaymath}
Note that the Koopman eigenfunctions for this system coincide with the Fourier functions on the circle.  

We consider that we observe the system through a binary measurement function $ h : M \to \mathbb{R} $, with
\begin{displaymath}
h  = 1_{M_1}, \quad M_1 = [ 0, \alpha ), \quad \alpha \in ( 0, 2 \pi ).
\end{displaymath}
We also define $ M_0 = M_1^c = [ \alpha, 2 \pi ) $. For this choice of observation map, the associated multiplication operator $ A = T_h \in B(L^2(\mu)) $ has pure point spectrum, $ \sigma( A ) = \sigma_p(A) =  \{ a_0, a_1 \} $, where $ a_0 = 0 $ and $ a_1 = 1 $. Moreover, the orthogonal projection operators to the corresponding eigenspaces, respectively denoted by $ H_0 $ and $ H_1$, are given by $\proj_{H_i} = T_{1_{M_i}} $. The spectral measure $E_h : \mathcal{B}(\mathbb{R}) \to B(L^2(\mu)) $ associated with $A$ is then given by
\begin{align*}
E_h(\Omega) &= 1_\Omega(a_0) \proj_{H_0} + 1_\Omega(a_1) \proj_{H_1} \\ 
&= 1_\Omega(a_0) T_{1_{M_0}} + 1_\Omega(a_1) T_{1_{M_1} },
\end{align*}
and we also have
\begin{displaymath}
    A = \int_\mathbb{R} a \, dE_h(a) = a_0 \proj_{H_0} + a_1 \proj_{H_1} = \proj_{H_1}.
\end{displaymath}

We now examine how (i) the state and measurement probability evolve between measurements under the unitary Koopman operators; and (ii) how the state is updated when measurements take place under projective dynamics. Working throughout in the Koopman eigenfunction basis $ \{ \phi_j \} $, we begin by computing the matrix elements of the state $ \rho_t = U^{t*} \rho_0 U^t $  reached after dynamical evolution for time $ t $ starting from a state $ \rho_0 \in B_1(L^2(\mu))$, in accordance with step DA2:
\begin{align}
\nonumber\rho_{t,jk} &= \langle \phi_j, \rho_t \phi_k \rangle_\mu = \langle U^t \phi_j, \rho_0 U^t \phi_k \rangle_\mu \\
\label{eqUCircle}&= e^{i(k-j)\omega t} \langle \phi_i, \rho_0, \phi_j \rangle_\mu = e^{i(k-j) \omega t } \rho_{0,jk}, 
\end{align}      
where $ \rho_{0,jk} = \langle \phi_j, \rho_0 \phi_k \rangle_\mu $. Next, we compute the matrix elements of the spectral projectors $E_h(\{a_i\}) $ in the Koopman eigenfunction basis, i.e.,  
\begin{displaymath}
    E_{i,jk} := \langle \phi_j, E_h(\{a_i\}) \phi_k \rangle_\mu = \langle \phi_j, \proj_{H_i} \phi_k \rangle_\mu,
\end{displaymath}
where
\begin{align*}
    E_{0,jk} &= 
\begin{cases}
    1-\frac{\alpha}{2\pi}, &j=k,\\
    \frac{-1}{(k-j)\pi} e^{i(k-j)\alpha/2} \sin\left(\frac{(k-j)\alpha}{2}\right), & j \neq k,
\end{cases}\\
E_{1,jk} &=
\begin{cases}
    \frac{\alpha}{2\pi}, &j=k,\\
    \frac{1}{(k-j)\pi}e^{i(k-j)\alpha/2} \sin\left(\frac{(k-j)\alpha}{2}\right), & j \neq k.
\end{cases}
\end{align*}
Using these formulas, the probability $ P_i(t) $ for a measurement of $A$ to take value $a_i$ at time $t$, starting from state $\rho_0$, and assuming no intervening measurements, is given by (DA4),    
\begin{equation}
    \label{eqPCircle}
    P_i(t) = \tr(E_h(\{a_i\}) \rho_t ) = \sum_{j,k=-\infty}^\infty E_{i,jk} \rho_{t,kj}.
\end{equation}
Moreover, the state $\rho^+_i$ immediately after a measurement $a_i$ of $A$ has been observed, and the system was in state $ \rho^-$ right before the measurement, has matrix elements (DA5) 
\begin{align}
    \nonumber\rho^+_{i,jk} &= \langle \phi_j, \rho^+_i \phi_k \rangle_\mu  
    =   \frac{\langle \phi_j,E_h(\{a_i\}) \rho^- E_h(\{a_i\}) \phi_k \rangle_\mu}{Z_i} \\ 
    \label{eqDACircle} &= \sum_{l,m=-\infty}^\infty \frac{E_{i,jl}\rho^-_{lm} E_{i,mk}}{Z_i},
\end{align}
where $ \rho^-_{lm} = \langle \phi_l, \rho^- \phi_m \rangle_\mu$ and
\begin{displaymath}
Z_i = \tr(E_h(\{a_i\})\rho^-E_h(\{a_i\})) = \sum_{j,l,m=-\infty}^\infty E_{i,jl}\rho^-_{lm} E_{i,mj}.
\end{displaymath}

To perform data assimilation in practice using the expressions derived above, we choose a spectral resolution parameter $L \in \mathbb{N}_0$, and approximate all operators by composing them by orthogonal projections $\Pi_L : L^2(\mu) \to L^2(\mu)$, mapping into the $(2L+1)$-dimensional subspace spanned by $ \phi_{-L},\ldots, \phi_L$. That is, we approximate $\rho_t $ in DA2, $P_i(t)$ in~\eqref{eqPCircle}, and $\rho^+_i$ in~\eqref{eqDACircle} by
\begin{equation}
   \label{eqQMDAApprox}
   \begin{aligned}
        \hat \rho_{t} &= \frac{U^{t*}_L \rho_0 U^t_L}{\tr(U^{t*}_L \rho_0 U^t_L)}, \\
        \hat P_i(t) &= \tr( E_{h,L}(\{a_i\}) \hat \rho_t) = \sum_{j,k=-L}^L E_{i,jk} \rho_{t,kj},\\
        \hat \rho^+_{i} &=  \frac{E_{h,L}(\{a_i\} ) \rho^-  E_{h,L}(\{a_i\}) }{\tr({E_{h,L}(\{a_i\} ) \rho^-  E_{h,L}(\{a_i\}) })},
    \end{aligned}
\end{equation}
respectively, where $ U^t_L = \Pi_L U^t \Pi_L $, and $E_{h,L}(\Omega) = \Pi_L E_h(\Omega ) \Pi_L $ for any Borel set $\Omega \subseteq \mathbb{R}$. In particular, $ \hat \rho_t$ and $ \hat \rho^+_i $ have matrix elements
\begin{align*}
    \hat \rho_{t,jk} &= \langle \phi_j, \hat \rho_t \phi_k \rangle_\mu = \frac{e^{i(k-j)\omega t} \rho_{0,jk}}{ \sum_{p,q=-L}^L e^{i(q-p)\omega t} \rho_{0,pq} }, \\ 
    \hat \rho^+_{i,jk} &= \langle \phi_j, \hat \rho^+_i \phi_k \rangle_{\mu} = \frac{\sum_{l,m=-L}^LE_{i,jl} \rho^{-}_{lm} E_{i,mk}}{\sum_{p,q,r=-L}^L E_{i,pq} \rho^-_{qr} E_{i,rp} },
\end{align*}
respectively. Note that the division by $ \tr(U^{t*}_L \rho_0 U^t_L) $ in the expression for $ \hat \rho_t $ is due to the fact that, unlike $U^t$, $U^t_L$ is not unitary, and thus does not preserve the trace of $\rho_0$. Since all operators involved are bounded, and thus continuous, linear operators, the expressions above converge as $L\to\infty$. 

Figure~\ref{figPCircle} displays the evolution of the probability $ \hat P_1(t) $ for a measurement $ a_1 = 1 $ to occur, computed via this approach for three different choices of $ \alpha $ (controlling the relative size of the subsets $ M_i \subset M $ on which $h$ takes values $a_i$) and the time interval between observations, denoted $ \Delta t $. All experiments start at time $ t= 0 $ from the stationary state $ \bar \rho $ (see Section~\ref{secQMDA}), which corresponds to a probability $ P_1(0) = \alpha / 2 \pi $ to observe $ a_1$. Moreover, the initial state $ \theta_0 \in S^1 $ in state space has phase angle equal to 0, and we use the spectral resolution parameter $L=64$.  In Figs.~\ref{figPCircle}(a, b) and \ref{figPCircle}(c), we set $ \alpha = \pi $ and $ \pi / 6 $, respectively. In the former two cases, this results in equal probability to observe $0$ and $1 $ with respect to the invariant measure, which is manifested by the ``truth'' time series $ h(t) = h(\Phi^t(\theta_0)) $ exhibiting a regular square waveform. On the other hand, in Fig.~\ref{figPCircle}(c), $h(t) $ has an intermittent character, as the probability for $ h $ to take value $1$ is six times smaller than the probability for it to take value $ 0 $.  In all three cases, the observation time interval $ \Delta t $  is set to an irrational multiple of the rotation period, $ T = 2 \pi / \omega$; specifically,  $ \Delta t = q T / ( 50 \sqrt{2}) \approx 0.014 q$, with $ q = 20 $ in Fig.~\ref{figPCircle}(a) and 200 in Figs.~\ref{figPCircle}(b, c). Thus, Figs.~\ref{figPCircle}(a) and \ref{figPCircle}(b, c) correspond to frequent and infrequent observations relative to the rotation period, respectively.

\begin{figure}
    \centering
    \includegraphics[width=.92\linewidth]{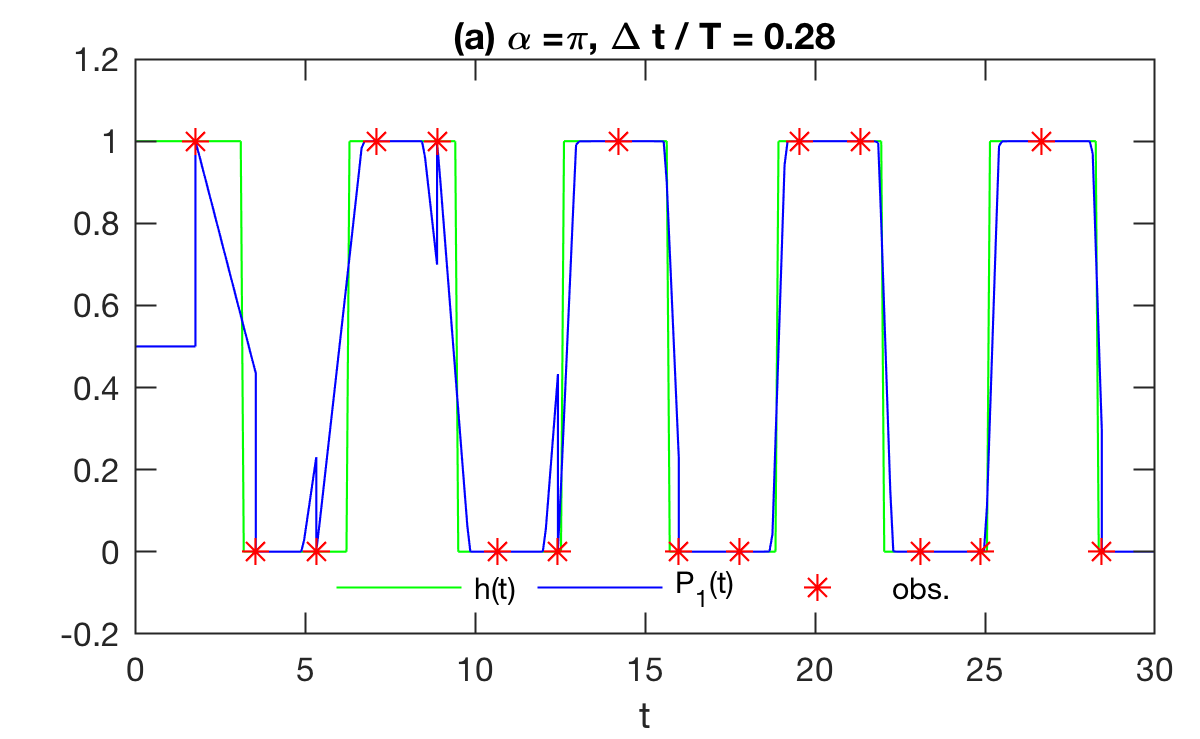}
    \includegraphics[width=.92\linewidth]{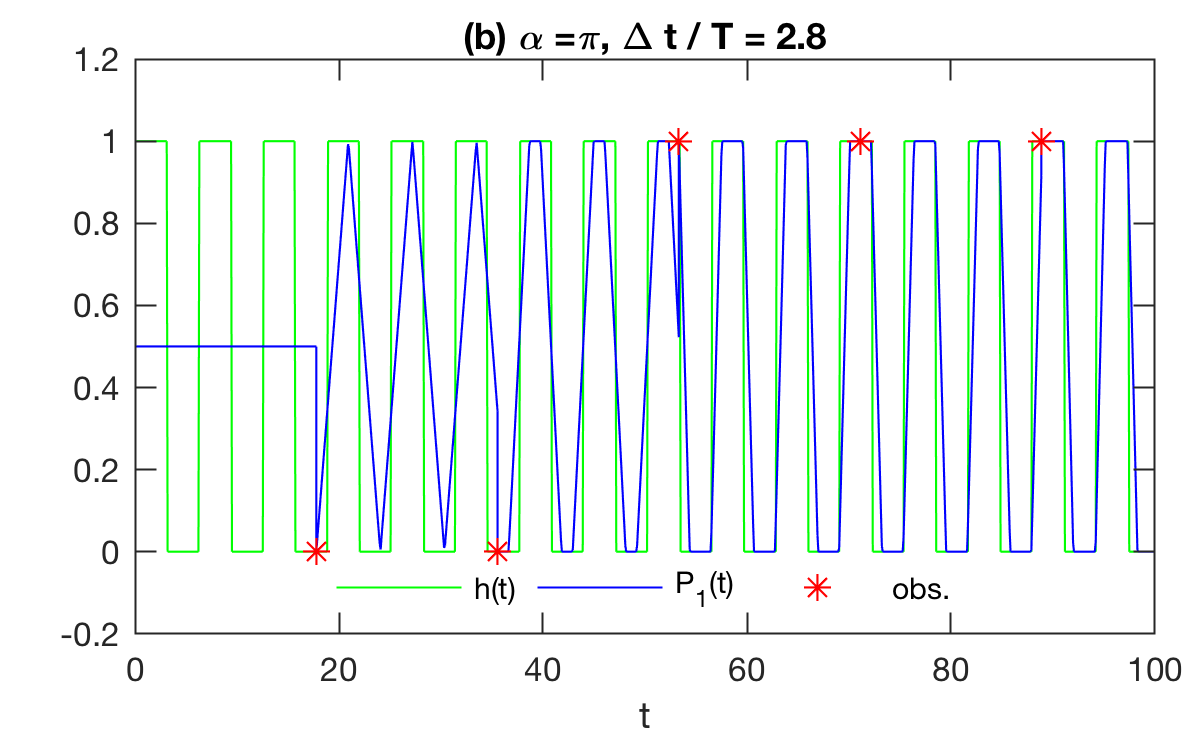}
    \includegraphics[width=.92\linewidth]{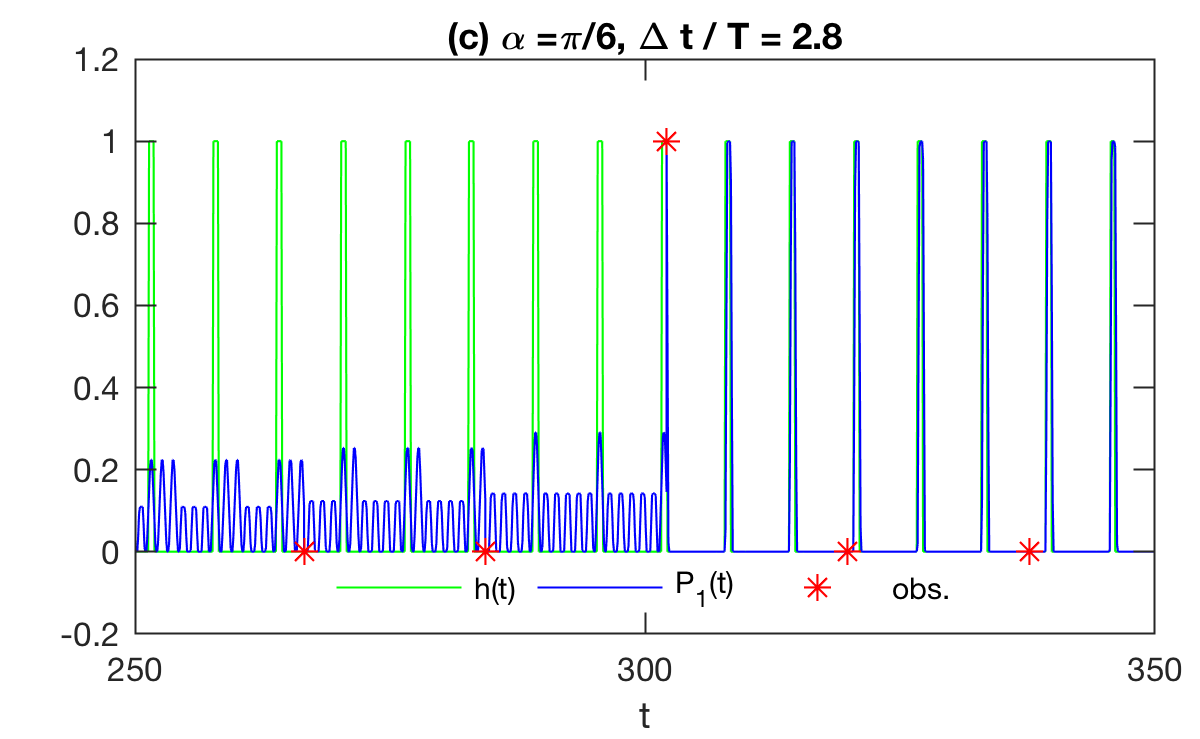}
    \caption{\label{figPCircle}Evolution of the measurement probability $ \hat P_1(t) $ for the binary observable $ T_h $ of the circle rotation, determined via QMDA. Three cases are shown, the first two of which (a, b) have $ \alpha = \pi $ (i.e., equal probability to observe 0 and 1 with respect to the invariant measure), with frequent (a) and infrequent (b) observations relative to the rotation period $ T = 2 \pi / \omega = 2 \pi $. Case (c) has $ \alpha = \pi/6 $ (i.e., the stationary probability to observe 1 is $1/6$ of the probability to observe 0) and infrequent observations as in (b).  The true signal $ h(t) $ and observations are also shown in each panel for reference. In (c), the first 250 time units of the data assimilation period are omitted for clarity of visualization, but have no $a_1=1$ measurements. Notice the improvement of skill after an $a_1 $ measurement is made shortly after $ t = 300$.}
\end{figure}

In all three cases, following an initial transient period, whose length depends strongly on both $\alpha$ and $ \Delta t/ T$, the data assimilation system locks in a pattern for $ \hat P_1(t)$, which tracks the signal $h(t)$ essentially in a deterministic manner. That is,  $ \hat P_1(t) \approx 0 $ whenever $ h(t) = 0 $, and $\hat  P_1(t) \approx 1 $ whenever $ h(t) = 1 $. In Fig.~\ref{figPCircle}(a), accurate tracking of $ h(t) $ is seen to take place from approximately $ t = 14 $. In Fig.~\ref{figPCircle}(b), the time to attain accurate tracking increases to $ t \simeq 55 $ (due to infrequent observations), while in Fig.~\ref{figPCircle}(c) accurate tracking does not take place until after $t=300$ (due to both infrequent observations and low probability to observe $h(t)=1$). It is worthwhile noting that in both Figs.~\ref{figPCircle}(b, c) there is a marked increase in tracking accuracy after the first $ h(t)=1$ observation is made; this is particularly evident in Fig.~\ref{figPCircle}(c).

\section{\label{secCompactification}Spectral discretization of observables}

As stated in Section~\ref{secQMDA}, the state update formula in step DA5 is only applicable if the measurement $a$ lies in the point spectrum of the observable $A$.  In this section, we introduce a modification of that step, which renders it applicable for arbitrary bounded observables associated with multiplication operators. Specifically, we consider the case $ A = T_h $ with $ h : M \to \mathbb{R} $ a function in $L^\infty(\mu)$. 

Recall that  $ a \in \mathbb{R} $ lies in the point spectrum of $T_h $ if and only if the corresponding level set $ h^{-1}(\{a\}) \subseteq M $ has positive $\mu $  measure. As a result, the problematic measurement points $ a $ in the context of step DA5 are those with null corresponding level sets with respect to $ \mu$. These facts suggest that a possible remedy for the ill-definition of DA5 is to approximate $ h$ by a function $ \bar h : M \to \mathbb{R} $ whose level sets have all positive $ \mu $ measure. Because $ \mu$ is a probability measure, $ \bar h $ would necessarily take countably many values on a full-measure subset of $ M $. Here, we will in fact construct $\bar h $ so as to take finitely many values through an averaging procedure applied to $h$, which we now describe. 

\subsection{\label{secCondAv}Conditional averaging}

Let $ \cdf_h : \mathbb{R} \to [ 0, 1 ] $ be the cumulative distribution function (CDF) of $h$, defined as 
\begin{displaymath}
    \cdf_h(a) = \mu( \{ x \in M : h(x) \leq a \} ). 
\end{displaymath}
Any finite partition $ \{ J_0, \ldots, J_{S-1} \} $ of $ ( 0, 1 ) $ into intervals $ J_i \subseteq [ 0, 1 ] $ of equal length, $ 1/ S $, induces partitions $ \Xi = \{ \Xi_0, \ldots, \Xi_{S-1} \} $ and $ \mathcal{M} = \{ M_0, \ldots, M_{S-1} \} $ of $ \mathbb{R} $ and $M$, respectively, whose elements $ \Xi_i = \cdf_h^{-1}(J_i)$ and $ M_i = h^{-1}( \Xi_i )$ have equal measure, $ \mu_h( \Xi_i ) = \mu( M_i ) = 1 / S $, by construction. Here, $ \cdf_h^{-1} : (0, 1) \to \mathbb{R}$ is the quantile function of $h$, defined as
\begin{displaymath}
    \cdf_h^{-1}( s ) = \inf \{ a \in \mathbb{R} : \cdf_h( a) \geq s \}. 
\end{displaymath}
Let also $ \pi_h : \mathbb{R} \to \{ 0, \ldots, S-1\} $  be the affiliation function (projection map) associated with the partition $ \Xi $, mapping $ a \in \mathbb{R} $ to the index $ i$ of the unique set $ \Xi_i \in \Xi  $ in which $ a $ lies. Similarly, define the affiliation function $ \pi : M \to \{ 0, \ldots, S-1\} $ for $ \mathcal{M} $, where $ \pi = \pi_h \circ h $.  We approximate $h$ by its conditional expectation, $ \bar h : M \to \mathbb{R}$, conditioned on the affiliation function $ \pi$, viz.
\begin{displaymath}
    \bar h = \mathbb{E}( h \mid \pi ) =  \sum_{i=0}^{S-1} \bar a_i 1_{M_i}, \quad \bar a_i = \int_{M_i} h \, d\mu.
\end{displaymath}
It then follows that the corresponding multiplication operator $ T_{\bar h} \in B(L^2(\mu))$ has pure point spectrum, and is characterized by the purely atomic projection-valued measure $ E_{\bar h} : \mathcal{B}(\mathbb{R}) \to B(L^2(\mu))$ satisfying
\begin{equation}
    \label{eqPVMQuant}
    E_{\bar h}(\{ \bar a_i \}) = E_{\bar h}(\Xi_i) =  T_{1_{M_i}}.
\end{equation}
With these definitions, we replace step DA5 with the following:
\begin{enumerate}[{DA}1$'$), wide]
\setcounter{enumi}{4}
\item If the data assimilation state immediately before a measurement is $ \rho^- \in B_1(L^2(\mu)) $, and a measurement of $ T_h $ yields the value $ a \in \sigma(T_h)$,  then the state $ \rho^+ $ immediately after the measurement is given by
\begin{displaymath}
    \rho^+ = \frac{ E_{\bar h}(\{ \bar a_i\}) \rho^- E_{\bar h}(\{ \bar a_i\}) }{ \tr (E_{\bar h}(\{ \bar a_i\}) \rho^- E_{\bar h}(\{ \bar a_i\}) ) }, \quad i = \pi_h( a ).
\end{displaymath}
\end{enumerate}

Note that despite this modification of DA5, the measurement probabilities $P_i(t) $ in step DA4, evaluated with respect to $ T_h $ on the elements $ \Xi_i $ of the partition, are consistent with the measurement probabilities with respect to the quantized observable $T_{\bar h} $ on the same set, i.e., for any state $ \rho_t \in B_1(L^2(\mu))$, 
\begin{equation}
    \label{eqPQuantized}
    P_i(t) = \tr(E_h(\Xi_i)\rho_t) = \tr( E_{\bar h}(\{ \bar a_i \}) \rho_t). 
\end{equation}

\subsection{\label{secRelEnt}Information-theoretic measures of skill}

To assess the skill of QMDA, we use relative-entropy measures associated with the partition $\Xi$  \cite{GiannakisEtAl12b}. In particular, at any given time $t$, associated with this partition are three discrete probability measures on $\mathbb{R}$, namely (i) the equilibrium measure $ \bar \nu(Z)= \sum_{i: Z \cap J_i \neq \emptyset} 1/S $ induced by the invariant measure of the dynamics; (ii) the measure $ \nu_t(Z) = \sum_{i:Z \cap J_i \neq \emptyset} P_i(t) $ associated with the data assimilation probabilities from~\eqref{eqPQuantized}; and (iii) the measure $ \tilde \nu_t(Z) = \sum_{i: Z \cap J_i \neq \emptyset} 1_{J_i}( \pi_h(h(t)) ) = \delta_{\pi_h(h(t))}(Z) $ associated with the true signal $ h(t) =  h(\Phi^t(\theta_0))$. Here, $ Z $ is an arbitrary Borel subset of $\mathbb{R} $, and $ \delta_b $ the Dirac measure supported at $ b \in \mathbb{R} $. Using these probability measures, we compute the relative entropies 
\begin{align*}
    \mathcal{D}(t) &= D_\text{KL}( \nu_t \mid\mid \nu ) = \sum_{i=0}^{S-1} P_i(t ) \log_2( S P_i(t) ), \\
    \mathcal{E}(t) &= D_\text{KL}( \tilde \nu_t \mid\mid \nu_t ) = - \log_2 P_{\pi_h( h(t) )}( t ),
\end{align*}
where $ D_\text{KL}( \cdot \mid\mid \cdot) $ denotes relative entropy (Kullback-Leibler divergence) between discrete probability distributions. 

The quantities $ \mathcal{D}(t) $ and $ \mathcal{E}(t) $ are information-theoretic measures of the precision and ignorance of the data assimilation distribution $ \nu_t $. Specifically, $ \mathcal{D}_t $ measures the information content of $ \nu_t $ beyond the equilibrium measure $ \bar \nu $ (which can be thought of as a null hypothesis), while $ \mathcal{E}_t $ measures the lack of information of $ \nu_t $ relative to the truth distribution $ \tilde \nu_t $. The fact that we use base-2 logarithms in our definition of relative entropy means that these information gain/losses are measured in ``bits''. Note that it follows from standard properties of relative entropy that $ \mathcal{D}(t) $ is non-negative, vanishes if and only if $ \nu(t) = \bar \nu $, and is bounded above by $ \log_2S $. The latter, is equal to $ D_\text{KL}( \tilde \nu_t \mid \mid \bar \nu ) $. $ \mathcal{E}(t)$ is similarly non-negative, and vanishes if and only if $ \nu_t = \tilde \nu_t $.  Thus, a ``perfect'' data assimilation scheme would attain $ \mathcal{D}(t) = \log_2 S $ and $ \mathcal{E}(t) = 0$.  Unlike $ \mathcal{D}(t)$,  $ \mathcal{E}(t) $ is unbounded, but the value $ \log_2 S$ happens to also be equal to $ D_\text{KL}( \bar \nu \mid\mid \tilde \nu_t ) $, so that data assimilation distributions with $ \mathcal{E}(t) > \log_2 S $ have more ignorance relative to the truth than the equilibrium measure. As a result, $ \mathcal{E}(t) < \log_2 S$ and $ \mathcal{E}(\tau) \geq \log_2S $ are natural criteria to distinguish between useful versus non-useful data assimilation predictions, respectively.      

\subsection{Application to the circle rotation}

As a demonstration of the approaches in Sections~\ref{secCondAv} and~\ref{secRelEnt}, consider again the periodic dynamical system from Section~\ref{secCircle}, now observed via the continuous observation map $ h : M \to \mathbb{R} $ with $ h(\theta) = \cos \theta $. For this choice of observation map, $T_h$ has purely continuous spectrum, and 
\begin{displaymath}
    \cdf_h(a) = 1 - \frac{\cos^{-1}(a)}{\pi}, \quad \cdf_h^{-1}(b) = \cos( (1-b) \pi ).
\end{displaymath}
It thus follows that for any interval $ J_i \in \{ J_0, \ldots, J_{S-1} \} $ with endpoints $b_i < b_{i +1}$, $ b_i = i / S $, 
\begin{align*}
    \Xi_i &= [ \cos( ( 1 - b_i ) \pi ), \cos( ( 1 - b_{i+1} ) \pi ) ), \\
    M_i &= ( ( 1 - b_{i+1} ) \pi, ( 1 - b_i ) \pi ] \cup ( ( b_i - 1 ) \pi, ( b_{i+1}-1) \pi ], \\
    \bar a_i &= \frac{\sin( ( 1 - b_{j+1})\pi ) - \sin( (1-b_j)\pi)}{\pi}.
\end{align*}
Using the above, we can compute formulas for the matrix elements $E_{i,jk} = \langle \phi_j, E_{\bar h}(\{ \bar a_i \}) \phi_k \rangle_\mu$ in the Koopman eigenfunction basis, namely,
\begin{displaymath}
    E_{i,jk} = 
    \begin{cases}
        b_{i+1} - b_i, & j = k, \\
        \frac{\sin( (k-j)( 1 - b_i) \pi) - \sin( (k-j)(1 - b_{i+1})\pi)}{(k-j)\pi}, & j \neq k.
    \end{cases}
\end{displaymath}
The above, in conjunction with the expressions in~\eqref{eqUCircle} and~\eqref{eqPCircle} for the evolution of the state and measurement probabilities are sufficient to carry out our data assimilation scheme.

Figures~\ref{figPCircleX} and~\ref{figPCircleX2} show results for the measurement probabilities $P_i(t)$ and the relative-entropy metrics $\mathcal{D}(t)$ and $\mathcal{E}(t)$,   obtained for the circle rotation from Section~\ref{secCircle} with frequency $ \omega = 2\pi/ T = 1 $ and measurement interval $ \Delta t =200 T / ( 50 \sqrt{2} ) \approx 2.8 T $ (i.e., the infrequent-observations case from Section~\ref{secCircle} and Fig.~\ref{figPCircle}(c)), using a partition of $ S = 32 $ elements and a spectral resolution of $ L = 64 $ for operator approximation. As in Section~\ref{secCircle}, the experiment starts from the stationary state $ \bar \rho $ (setting again the initial state of the underlying system to $\theta_0 = 0$). Correspondingly, until the first measurement is made at $ t \approx 2.8 T $, the measurement probability is uniform, $ P_i(t) = 1 / S \approx 0.03$, the precision metric is zero, $\mathcal{D}(t) = 0$, and the ignorance metric is equal to the number of bits in the partition, $ \mathcal{E}(t) =  \log_2 S = 5 $. 

When the first measurement is made, $P_i(t)$ collapses to a strongly bimodal distribution, consistent with the fact that $h(\theta) = \cos(\theta) $ is a two-to-one function on the circle. Note, in particular, that in Figs.~\ref{figPCircleX}(a) and~\ref{figPCircleX2}(a) one of the two branches of the measurement probability distribution accurately tracks the true signal $ h(t)$, but on the basis of a single measurement, the data assimilation system assigns nearly equal probability to the two branches. The increase of skill following the first measurement is also manifestly visible in the relative-entropy plots in Fig.~\ref{figPCircle}(c), where $\mathcal{D}(t)$ is seen to jump to $ \simeq 3.5 $ upon occurrence of the first measurement. At that time, the ignorance metric $\mathcal{E}(t)$ exhibits an appreciable decrease from $ \log_2 S $, but is seen to undergo intermittent excursions to $ \geq \log_2 S $ values. Closer inspection (Figs.~\ref{figPCircleX}(b) and~\ref{figPCircleX2}(a, b)) indicates that these excursions are likely due to phase alignment errors between $P_i(t) $ and $ h(t)$. 

Next, as soon as the second measurement arrives, the measurement probability collapses to a unimodal distribution that accurately tracks the true signal. This contrasts the behavior seen in Figs.~\ref{figPCircle}(b,c), where, due to the lower discriminating power of the binary observable employed there, multiple measurements are required before the data assimilation system accurately tracks $h(t)$. With successive measurements, the phase alignment error seen at early times gradually diminishes, and by $ t \simeq 500 $, the measurement probabilities $ P_i(t) $ track the truth signal with persistently high precision and low ignorance (see Figs.~\ref{figPCircleX}(d) and~\ref{figPCircleX2}(c)).  

\begin{figure}
\includegraphics[width=\linewidth]{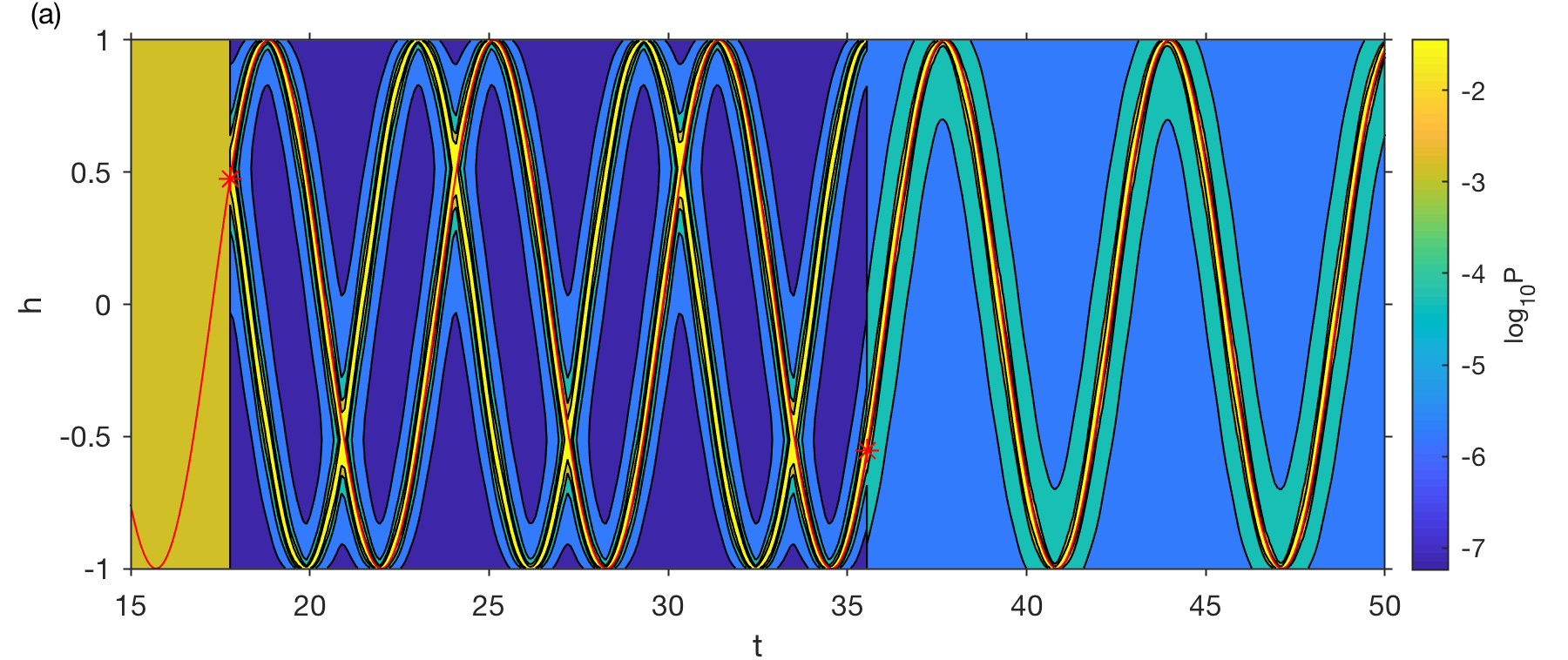}
\includegraphics[width=\linewidth]{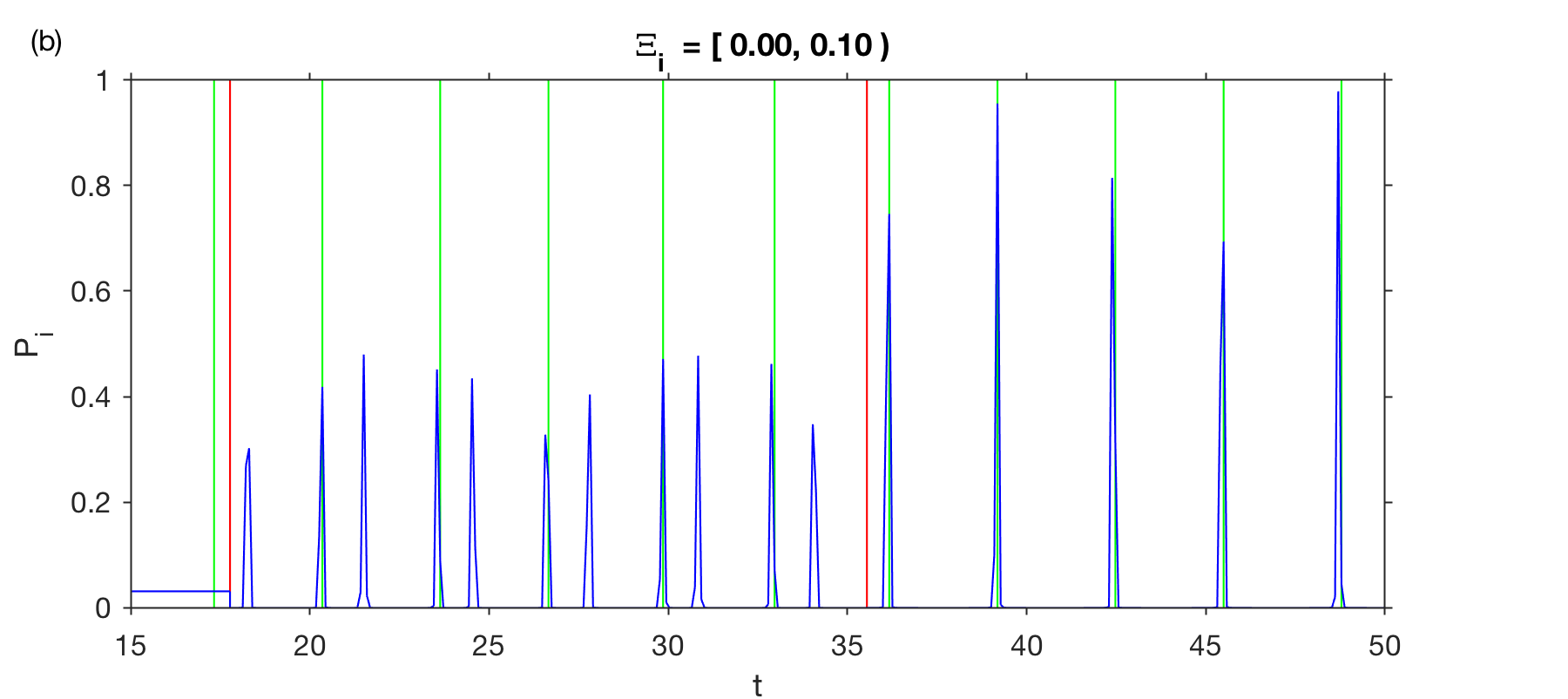}
\includegraphics[width=\linewidth]{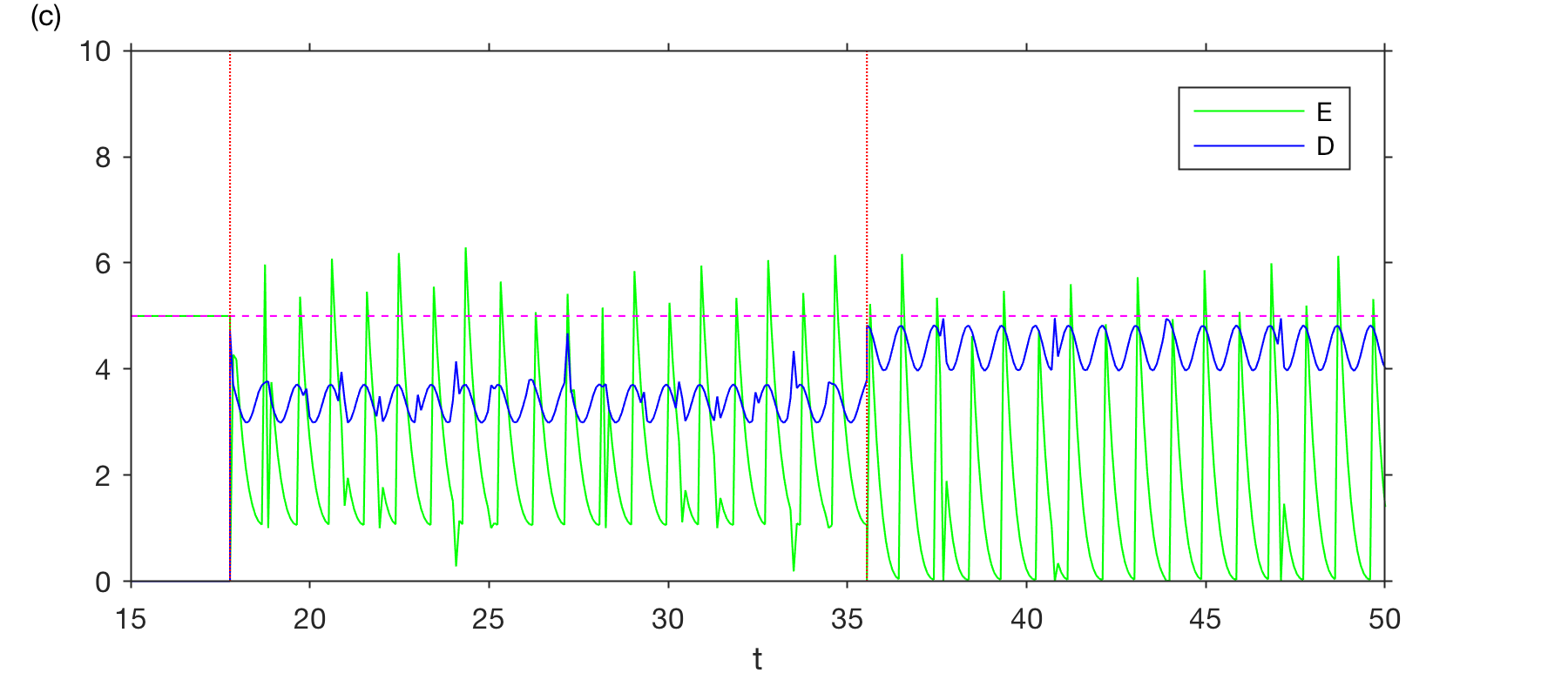}
\includegraphics[width=\linewidth]{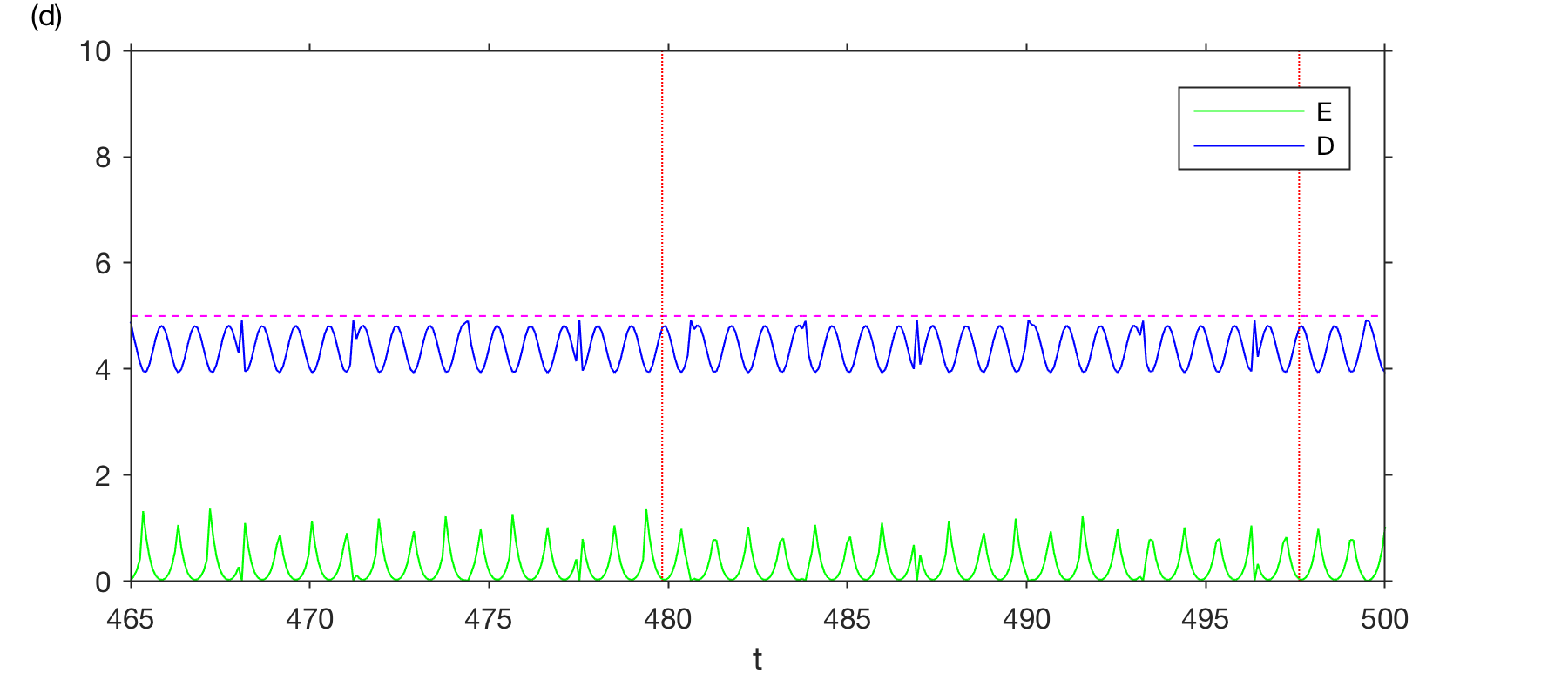}
\caption{\label{figPCircleX}Results of QMDA applied to observable $h(\theta) = \cos\theta $ of the periodic dynamical system on the circle for an observation interval $ \Delta t \approx 2.8 T$ and initial state $ \bar \rho $  as in Fig.~\ref{figPCircle}(c).  (a) Logarithm of the measurement probability $ \hat P_i(t)$ for $h$ to take values in a partition $ \Xi $ of $\mathbb{R}$, consisting of $ S = 32 $ elements of equal probability mass with respect to the invariant measure $ \mu $. The time interval shown contains the first two observations, indicated by red asterisks. The true signal $ h(t) $ is shown in a red line for reference. (b) Measurement probability $ \hat P_i(t)$ for element $ \Xi_{17} \approx [ 0.00, 0.10 ) $ of the partition. The thin shaded grid regions indicate time intervals where $h(t) $ takes values in $ \Xi_{17} $, and vertical red lines indicate observation time instances. (c, d) Precision and ignorance metrics, $ \mathcal{D}(t) $ and $ \mathcal{E}(t) $, for (c) the time interval shown in (a, b) and (d) a later time interval. Red vertical and magenta horizontal lines indicate observation time instances and the maximal number of bits, $ \log_2 S = 5 $, associated with the partition, respectively.  Observe the gradual decrease of $ \mathcal{E}(t) $ caused by misassignment of the occupancy times of the elements of $ \Xi $ at early times. This effect is visible upon close inspection of the $P_i(t) $ plot in (b), and better visualized in the contour plots in Fig.~\ref{figPCircleX2}.}
\end{figure}

\begin{figure}
    \includegraphics[width=\linewidth]{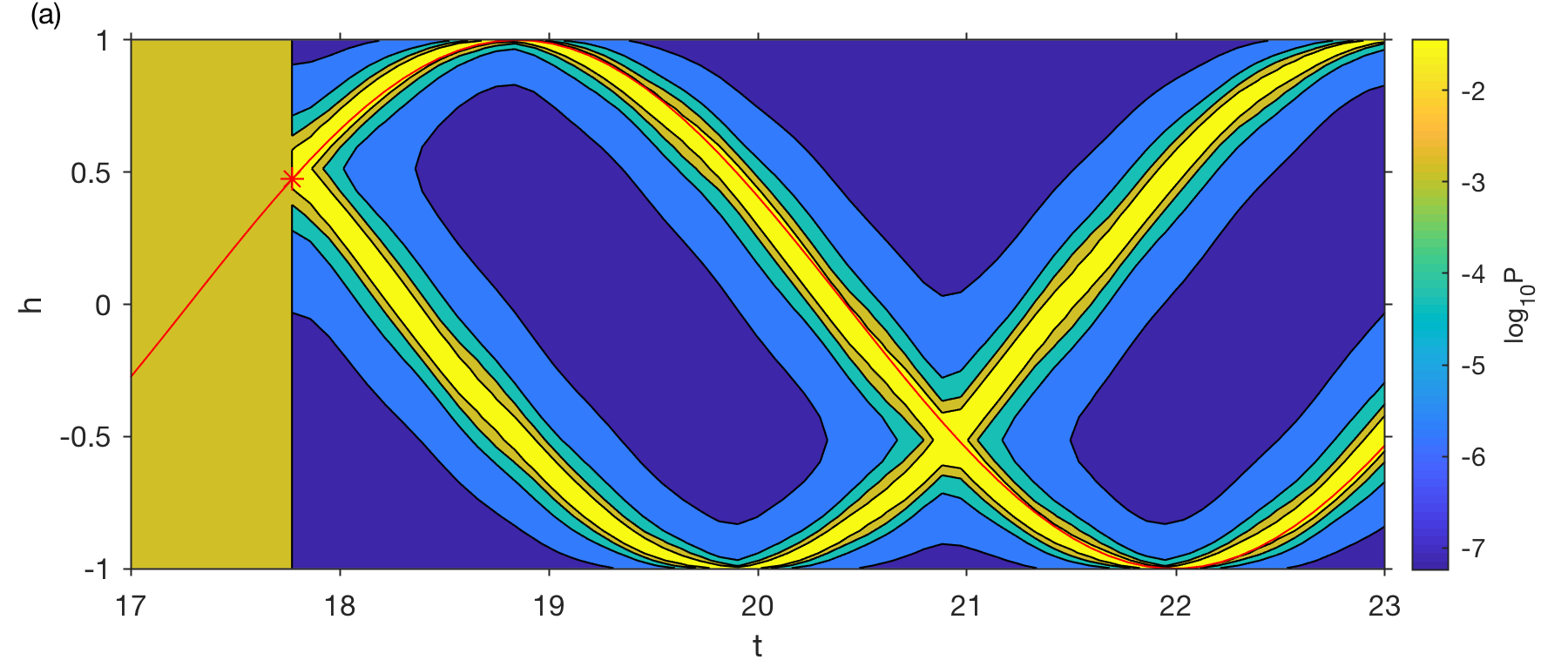}
    \includegraphics[width=\linewidth]{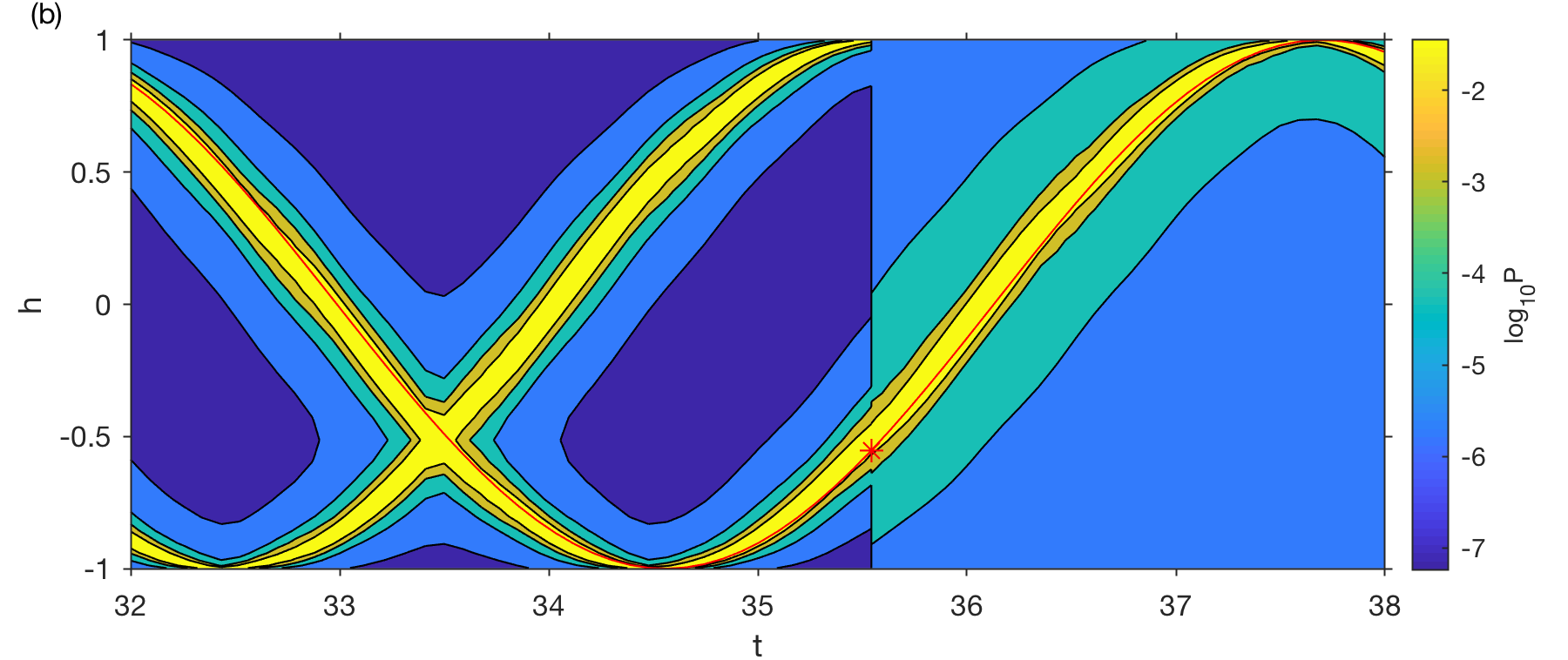}
    \includegraphics[width=\linewidth]{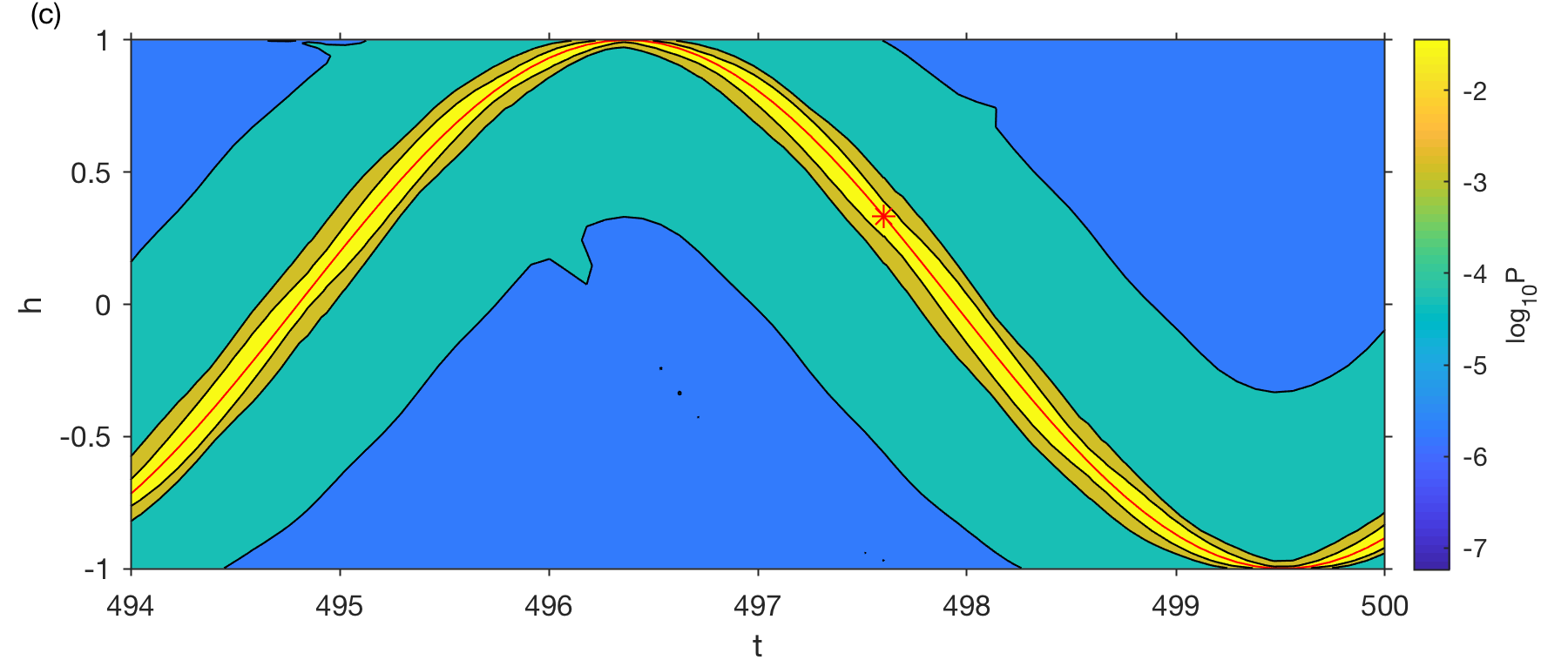}
    \caption{\label{figPCircleX2}(a, b) Detailed views of the evolution of the measurement probability $ \hat P_i(t)$ in Fig.~\ref{figPCircleX}(a) for time intervals containing the first (a) and second (b) measurements. Observe the collapse of the uniform distribution to a bimodal distribution in (a) and the subsequent collapse of the latter to a unimodal distribution in (b). (c) Detailed view of the measurement probability for a later time interval. Notice the improved alignment of the maximal $ \hat P_i(t)$ with the true signal $ h(t) $ compared to (a, b).}
\end{figure}

\section{\label{secDataDriven}Data-driven approximation}

The data assimilation framework presented thus far operates under the assumptions that (i) an orthonormal basis $ \{ \phi_j \} $ for the Hilbert space $L^2(\mu)$ associated with the invariant measure is available; and (ii) the action of the Koopman operators $U^t$ and spectral projectors $ E_{\bar h}(\Xi_i) $ on the basis elements can be computed so as to construct matrix representations of these operators. Arguably, this information will seldom be available in real-world applications, not least because $\mu$ is generally an unknown measure, supported on a non-smooth subset of state space $ M$ (e.g., a fractal attractor). Moreover, the equations of motion, allowing one in principle to act with the Koopman operator, are not known. In response, in this section we establish a data-driven formulation of QMDA, which employs a finite, time-ordered dataset consisting of observations of the system to (i) build an orthonormal basis for an appropriate Hilbert space approximating $L^2(\mu) $; and (ii) construct matrix representations  of operators approximating the Koopman operator and spectral projectors on that space. Convergence of the data-driven approximation schemes to their continuous counterparts will then follow in the limit of large data by ergodicity (see Appendix~\ref{appConvergence}). 

Our approach follows closely \cite{BerryEtAl15,GiannakisEtAl15,Giannakis17,DasGiannakis19,DasEtAl18,GiannakisEtAl19}, who employ kernel algorithms for statistical learning \cite{BelkinNiyogi03,CoifmanLafon06,BerryHarlim16} to build the basis through eigenfunctions of kernel integral operators obtained from the data. In what follows, we describe the main elements of this procedure, referring the reader to \cite{BerryEtAl15,GiannakisEtAl15,Giannakis17,DasGiannakis19,DasEtAl18,GiannakisEtAl19} for some of the mathematical details. Hereafter, $ X \subseteq M $ will denote the (compact) support of the invariant measure $ \mu $.     

\subsection{\label{secObs}Data-driven modeling scenario}

We consider that available to us is a time-ordered sequence $F(x_0), F(x_1), \ldots, F(x_{N-1}) $ of $N$ data points, sampled along a dynamical trajectory $ x_n = \Phi^{n\,\Delta t}(x_0) $, $ x_0 \in M $, through a continuous, injective observation map $ F: M \to Y $, taking values in a metric space $Y$ (the data space). Here, $ \Delta t $ is a positive sampling interval such that the discrete-time map $ \Phi^{\Delta t} : M \to M $ is also ergodic for the probability measure $ \mu $. In applications, the data space is typically linear and finite-dimensional, $ Y = \mathbb{R}^m$, but our methods also apply for nonlinear data spaces (e.g., directional data with $ Y = S^2$), or infinite-dimensional linear spaces (e.g., scalar-field, ``snapshot'' data). We will additionally assume that the observation function $ h : M \to \mathbb{R} $ is continuous, and its values $h(x_0), \ldots, h(x_{N-1}) $ on the sampled dynamical states are known.    

The observations $F(x_n) $ will be used below to construct the data-driven basis employed for operation approximation. In that context, the injectivity of $ F $ will be important to ensure completeness of the basis. In practical applications, the joint values $( F(x_n), h(x_n) ) $ could be acquired in an offline training phase where one has access to the full dynamical system on $M$. If access to an explicit injective map $ F $ is not available, but the values $ h(x_n) $ are still known, it is possible to employ an alternative approach, which involves building an injective map from $h $ through the use of delay-coordinate maps of dynamical systems \cite{Takens81,SauerEtAl91,Robinson05}. Specifically, given a nonzero integer parameter $Q$ (the number of delays), we define $ h_Q : M \to \mathbb{R}^{Q} $ with
\begin{equation}
    \label{eqHQ}
    h_Q(x) = \left( h( x ), h( \Phi^{-\Delta t}( x ) ), \ldots, h(\Phi^{-(Q-1)\, \Delta t}(x)) \right).  
\end{equation}
It is known that under mild assumptions on $ \Phi^t $,  $h $, and  $ \Delta t $, if $M$ is a finite-dimensional differentiable manifold, then for any compact set $\mathcal{U} \subseteq M $ there exists $ Q_* \in \mathbb{N} $ such that, for all $ Q > Q_* $, $ h_Q$ is injective on $ \mathcal{U} $ \cite{Takens81,SauerEtAl91}. Moreover, an analogous result holds if $M$ is a (potentially infinite-dimensional) Hilbert space, and $ \mathcal{U} \subset M $ is a compact subset of finite upper box-counting dimension, forward-invariant under $ \Phi^t $  \cite{Robinson05}.  

Together, the results in \cite{Takens81,SauerEtAl91,Robinson05} hold for many of the dynamical systems encountered in physical applications, including a broad range of ordinary differential equation and partial differential equation models. Noting, in particular, that $h_Q(x_n) $, $ Q-1 \leq n \leq N - 1$, can be evaluated given the time series $ h(x_0), \ldots, h(x_{N-1} ) $ without knowledge of the dynamical flow $ \Phi^t $, delay-coordinate maps provide a practical tool for implicitly constructing injective observation maps from partial (non-injective), time-ordered observations. As a result, in the absence of an explicit injective observation map $ F$, our approach will be to set $ F = h_Q $ with $Q$ sufficiently large.           

\subsection{\label{secSampling}Sampling measures and the associated $L^2$ spaces}

Associated with the dynamical trajectory $x_0, \ldots, x_{N-1}  $ is a sampling probability measure $ \mu_N = \sum_{n=0}^{N-1} \delta_{x_n}/ N $, consisting of equally weighted Dirac measures supported at the sampled states. Note that integration of a measurable function $ f : M \to \mathbb{C } $ with respect to $ \mu_N$ corresponds to a time average of its values at the sampled points, i.e., $ \int_M f \, d\mu_N = \sum_{n=0}^{N-1} f(x_n) / N $. In particular, $ \int_M f \, d\mu_N $ can be evaluated given the values $f(x_n) $ without explicit knowledge of the underlying dynamical states. A sequence of sampling measures $ \mu_N $ starting from a fixed state $ x_0 \in M $ is said to converge to the invariant measure $\mu $ weakly if for every bounded continuous function $ f : M \to \mathbb{C} $, $ \int_M f \, d\mu_N $ converges to $ \int_M f \, d\mu $ as $ N \to \infty$ (i.e., in the limit of large data). The set of all starting points $ x_0 \in M $ for which this property holds is said to be the basin of $ \mu $, and will be denoted by $ \mathcal{B}_\mu $. By ergodicity, $ \mu $-almost every point in the support $X$ of $ \mu $ lies in $ \mathcal{B}_\mu$; that is, $ \mathcal{B}_\mu $ is a full-measure set with $ \mu( \mathcal{B}_\mu) = 1 $. In fact, for many systems encountered in applications, $ \mathcal{B}_\mu $ is a significantly ``larger'' set than  $X$. For example, for systems that possess \emph{physical} measures \cite{Young02}, $\mathcal{B}_\mu $ has positive measure with respect to a reference ambient measure in state space (e.g., a Riemannian measure if $M$ is a Riemannian manifold). This means that the method will converge from a sufficiently large, experimentally accessible, set of initial conditions. 

Hereafter, we will always assume that $ \mu_N $ is a sampling measure associated with a dynamical trajectory starting in $\mathcal{B}_\mu$. By the assumptions stated above and time-continuity of the flow $ \Phi^t$, apart from the trivial case where $ \mu$ is a Dirac measure supported on a fixed point of the dynamics (which we will exclude by assumption), all states $ x_0, x_1, \ldots $ are distinct.  Besides these assumptions, an additional requirement we will make is that the dynamics has an absorbing ball property; specifically, we will require that the trajectory starting from any $ x_0 \in \mathcal{B}_\mu $ is contained within a compact subset  $ \mathcal{X} \subseteq M$, containing $X$. This assumption endows the space of continuous functions on $\mathcal{X}$, $C(\mathcal{X})$, with the structure of a Banach space (equipped with the uniform norm); this will be important for the convergence of the data-driven basis in Section~\ref{secEig}.

Next, as a data-driven analog of $L^2(\mu) $, we consider the Hilbert space $L^2(\mu_N) $ associated with the sampling measure $ \mu_N $. This space consists of equivalence classes $[f]_{\mu_N}$ of measurable functions $ f : M \to \mathbb{C}$ having common values at the sampled states $ x_0, \ldots, x_{N-1} $, and is equipped with the inner product $ \langle f, g \rangle_{\mu_N} = \int_M f^* g \, d\mu_N$. Because $ x_0,\ldots,x_{N-1} $ are all distinct points, $L^2(\mu_N) $ is an $N$-dimensional space isomorphic as a Hilbert space to $\mathbb{C}^N$, the latter equipped with a normalized Euclidean dot product, $ \vec f^\dag \vec g / N $. As a result, we can represent the $ L^2(\mu_N) $ equivalence class in which $ f : M \to \mathbb{C} $ lies by a column vector $ \vec f = ( f( x_0 ), \ldots, f(x_{N-1}) )^\top \in \mathbb{C}^N $, whose elements contain the values of $f$ at the sampled points. Moreover, we can represent every linear operator $ T : L^2(\mu_N) \to L^2(\mu_N) $ by a unique $N \times N $ matrix $ \bm T $ such that $ \vec g = \bm T \vec f $ is the column-vector representation of $ T [ f ]_{\mu_N}$. All of our data-driven techniques will utilize vectors and operators on $ L^2(\mu_N) $, so that they are readily implementable via the tools of matrix algebra; see Appendix~\ref{appComputational}.

\subsection{\label{secEig}Kernels and their associated eigenfunction bases}

We now describe how to build an orthonormal basis of $L^2(\mu_N)$ from the observed data $F(x_n)$ using kernel integral operators, and discuss the convergence of this basis to an orthonormal basis of $L^2(\mu)$ in the limit of large data. For the purposes of this work, a kernel will be a continuous, symmetric, positive-definite function $ k : M \times M \to \mathbb{R} $; that is, a continuous function with the properties that (i) $ k(x,x') = k(x',x) $ for all $x,x'\in M$; and (ii) for any finite sequence $ x_0, \ldots, x_{N-1} $ of points in $M$, the $N \times N $ matrix $ \bm K = [ k( x_m, x_n ) ] $ is positive-semidefinite. Given any Borel probability measure $ \nu $ on $M$ with compact support $ X_\nu $, the kernel $ k$ induces a self-adjoint, trace-class (thus compact) integral operator $ G_\nu : L^2(\nu) \to L^2(\nu) $, defined as
\begin{displaymath}
    G_\nu f = \int_M k( \cdot, x ) f(x) \, d\nu(x).
\end{displaymath}
In particular, there exists an orthonormal basis $ \{ \phi_0, \phi_1, \ldots \} $ of $L^2(\nu) $ consisting of eigenfunctions of $ G_\nu $ corresponding to non-negative eigenvalues $ \lambda_0, \lambda_1, \ldots $. By continuity of $k$ and compactness of $X_\nu $, every eigenfunction $ \phi_j $ with nonzero corresponding eigenvalue has a continuous representative $ \varphi_j \in C(M) $, such that
\begin{displaymath}
    \varphi_j(x) = \frac{1}{ \lambda_j} \int_M k(x,x') f(x') \, d\nu(x').
\end{displaymath}

The kernel $k$ will be said to be $L^2(\nu)$-strictly-positive if $ G_\nu $ is a positive operator, i.e., all eigenvalues $\lambda_j$ are strictly positive. In that case, all eigenfunctions $\phi_j$ have continuous representatives. Moreover $k$ will be called $L^2(\nu)$-Markov if $G_\nu$ is a Markov operator; i.e., $ Gf \geq 0 $ if $ f \geq 0 $, and $ Gf = f $ if $f$ is constant. $L^2(\nu)$-Markovianity implies, in particular, that the maximal eigenvalue $ \lambda_0 $ of $ G_\nu $ is equal to $1$, and there is a constant corresponding eigenfunction $ \phi_0 $, also equal to 1. An $L^2(\nu)$-Markov kernel will be said to be ergodic if $ \lambda_0 $ is a simple eigenvalue. We will use the symbol $ p : M \times M \to \mathbb{R} $ to distinguish a Markov kernel from a general kernel. 

Intuitively, the eigenbases $ \{ \phi_j \} $ associated with $L^2(\nu)$-strictly positive and Markov ergodic kernels can be thought of as generalizations of the Laplace-Beltrami eigenfunction bases associated with heat operators on Riemannian manifolds. In particular, if $ X_\nu$ had the structure of a smooth, closed Riemannian manifold, and $p$ was set to the heat kernel, the $ \phi_j $ would become Laplace-Beltrami eigenfunctions, which are well known to provide a smooth orthonormal basis for the $L^2$ space associated with the Riemannian measure \cite{Rosenberg97}.

Given a dynamical trajectory $ x_0, x_1, \ldots $ starting at $ x_0 \in \mathcal{B}_\mu$, with an associated forward-invariant compact set $\mathcal{X}$ and the corresponding sampling measures $ \mu_N $, $ N \in \mathbb{N}$, we will be interested in a family of kernels $ p_N : M \times M \to \mathbb{R} $ with the following properties: 
\begin{enumerate}
    \item $ p_N $ is a pullback kernel from data space; that is, there is a kernel $ \tilde p_N : Y \times Y \mapsto \mathbb{R} $ such that
        \begin{displaymath}
            p_N(x,x') = \tilde p_N( F(x), F(x')), \quad \forall x,x' \in M.
        \end{displaymath}
    \item $ p_N $ is $L^2(\mu_N)$-strictly-positive and Markov ergodic. 
    \item As $ N \to \infty $, the restriction of $ p_N $ to $\mathcal{X} \times \mathcal{X} $ converges uniformly to an $L^2(\mu) $-strictly-positive, Markov ergodic kernel $ p : \mathcal{X} \times \mathcal{X} \to \mathbb{R} $.  
\end{enumerate}
Property~1 above implies that the kernels $ p_N $ are data-driven, i.e., they can be evaluated at arbitrary states $x\in M $ from the corresponding observations $ F(x) \in Y$ alone. Property~2 implies that associated with the $ p_N $ is a Laplace-Beltrami-like, orthonormal basis $ \{ \phi_{N,0}, \ldots, \phi_{N,N-1} \} $ of $L^2(\mu_N) $ consisting of eigenfunctions $ \phi_{j,N} $ of $ G_{\mu_N} $ with continuous representatives $ \varphi_{j,N} \in C(M)$. In particular, this basis can be obtained from the eigenvectors $ \vec \phi_{j} $ of a known $N\times N $ kernel matrix $ \bm G = [ p_N(x_m,x_n) ]  $, where $ \vec \phi_j $ and $ \bm G $ represent $\phi_{j,N} $ and  $ G_{\mu_N}$, respectively, as described in Section~\ref{secSampling}. Under the assumptions stated in Sections~\ref{secObs} and~\ref{secSampling}, Property~3 implies that for every $ j \in \mathbb{N}_0$, in the limit of large data, $N \to \infty$, $ \varphi_{j,N}$ converges uniformly on $ \mathcal{X} $ to the continuous representative $ \varphi_j $ associated with an orthonormal basis $ \{ \phi_0, \phi_1, \ldots \} $ of $ L^2(\mu) $, consisting of eigenfunctions of $ G_\mu $. See \cite{DasEtAl18,GiannakisEtAl19} for proofs of these results, which make use of spectral convergence results for kernel integral operators established in \cite{VonLuxburgEtAl08}.     

Following \cite{DasEtAl18}, we construct the kernels $ \tilde p_N$ starting from an unnormalized kernel $ \tilde k_N : M \times M \to \mathbb{R} $, and applying to that kernel a normalization procedure to render it Markovian. Specifically, we set $k_N $ to the variable-bandwidth Gaussian kernel introduced in \cite{BerryHarlim16}, 
\begin{equation}\label{eqKVB}
    \tilde k_N( y, y' ) = \exp \left( - \frac{ d^2( y, y' ) }{ \epsilon \sigma_N(y) \sigma_N(y') } \right),
\end{equation}
and apply the symmetric (bistochastic) normalization proposed in \cite{CoifmanHirn13} to obtain $ \tilde p_N$. In~\eqref{eqKVB}, $ d: Y \times Y \to \mathbb{R} $, is a distance function, which we will nominally set to Euclidean distance (2-norm) for data in $ Y = \mathbb{R}^m$. Moreover, $ \epsilon $ is a positive parameter, tuned via an automatic procedure \citep[][Appendix~A]{BerryEtAl15}, and $ \sigma_N : Y \to \mathbb{R}_+ $ a continuous, positive-valued function whose role is to adaptively modify the localization of the kernel with respect to the sampling measure $ \mu_N $. In particular, it can be shown \cite{Giannakis17} that if the support $X$ has the structure of a smooth closed manifold, the corresponding $ \phi_{j,N} $ basis functions converge to Laplace-Beltrami eigenfunctions with respect to a Riemannian metric whose volume form has constant density relative to the invariant measure $\mu$ of the dynamics. While here we do not assume that $X$ has manifold structure (and thus cannot, in general, interpret the $ \phi_j$ as Laplace-Beltrami eigenfunctions), the balancing of the kernel localization due to $ \sigma_N$ plays an important role in enhancing the robustness of the data-driven basis to sampling errors.   

In what follows, we will employ the $ \phi_{j,N} $ basis of $L^2(\mu_N)$ obtained via this approach to formulate data-driven analogs of the QMDA framework described in Sections~\ref{secQMDA} and~\ref{secCompactification}. We refer the reader to \citep[][Algorithm~1]{DasEtAl18} for further details on the procedure to construct $ \tilde p_N $ and select the bandwidth parameter $ \epsilon $. Representative eigenfunctions $\phi_{j,N}$ obtained from data generated by the L63 system (to be studied in Section~\ref{secL63}) are displayed in Fig.~\ref{figL63Phi}.   

\begin{figure*}
    \includegraphics[width=\linewidth]{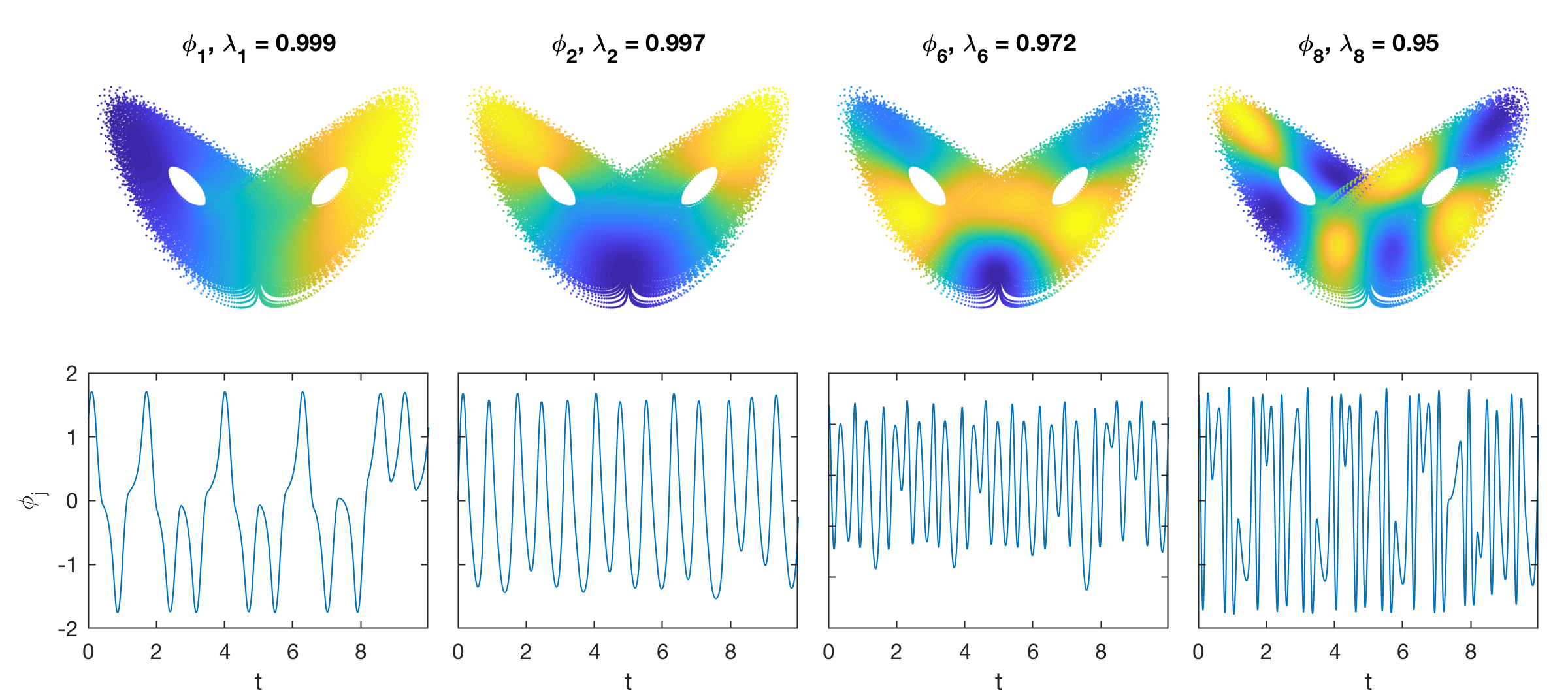}
    \caption{\label{figL63Phi}Representative data-driven eigenfunctions $\phi_{j,N}$ and their corresponding eigenvalues $\lambda_{j,N}$, computed from the fully observed L63 dataset in Section~\ref{secL63}. Top: Scatterplots of the eigenfunction values $ \phi_{j,N}(x_n) $ on the training dataset, with yellow (blue) colors corresponding to positive (negative) values. Bottom: Eigenfunction time series $t_n \mapsto \phi_{j,N}(x_n)$ over a portion of the training dataset spanning 10 natural time units. Notice that, despite the fact that the L63 attractor is not a Riemannian manifold, the eigenfunctions qualitatively resemble a Laplace-Beltrami eigenfunction basis with the corresponding heat-operator eigenvalues. That is, as $\lambda_{j,N}$ decreases, $\phi_{j,N}$ exhibits increasingly small-scale oscillatory behavior, allowing one to represent functions of increasingly fine structure through eigenfunction expansions.}  
\end{figure*}

\subsection{\label{secOpApprox}Operator approximation and convergence}

We now have the necessary ingredients to formulate a data-driven analog of the data-assimilation scheme presented in Sections~\ref{secQMDA} and~\ref{secCompactification}. Structurally, the data-driven formulation resembles closely its infinite-dimensional counterpart, with the Hilbert space $L^2(\mu) $ being replaced by $L^2(\mu_N) $ as described in Section~\ref{secSampling}, and the dynamical and measurement operators on $L^2(\mu) $ replaced by finite-rank operators on $L^2(\mu_N) $, as follows.
\begin{enumerate}[1., wide]
    \item The state $\rho \in B_1(L^2(\mu)) $ is replaced by a non-negative operator $ \rho_N \in B(L^2(\mu_N)) $ with $ \tr \rho_N = 1 $. In particular, the analog of the stationary state $ \bar \rho \in B_1( L^2(\mu)) $ is $ \bar \rho_N = \langle \phi_{N,0}, \cdot \rangle_{\mu_N} \phi_{N,0} $.  
    \item The Koopman operator $ U^t \in B(L^2(\mu))$ for $ t = q \, \Delta t $, $ q \in \mathbb{Z} $, is replaced by the $q$-step shift operator $U^{(q)} \in  B(L^2(\mu_N)) $, defined by 
        \begin{displaymath}
            U^{(q)}_N f( x_n ) = 
            \begin{cases}
                f( x_{n+q} ), & 0 \leq n \leq N-q-1, \\
                0, & N-q \leq n \leq N - 1.
            \end{cases}
        \end{displaymath}
    \item The CDF function $ \cdf_h$ employed in the construction of the partition $\mathcal{M}$ in Section~\ref{secCondAv} is replaced by the empirical CDF, $ \cdf_{h,N} : \mathbb{R} \to [ 0, 1 ] $, where 
        \begin{align*}
            \cdf_{h,N}(a) &= \mu_N( \{ x \in M : h(x) \leq a \}) \\
            &= \sum_{0 \leq n \leq N -1\; :\; h(x_n) \leq a} 1/ N.
        \end{align*}

        Given a uniform partition $ \{ J_0, \ldots, J_{S-1} \} $ of $(0,1)$, the empirical CDF induces partitions $ \Xi_N = \{ \Xi_{0,N}, \ldots, \Xi_{N-1,N} \} $ and  $\mathcal{M}_N = \{ M_{0,N}, \ldots, M_{S-1,N} \} $ of $ \mathbb{R} $ and  $M$, analogously to $ \Xi $ and  $\mathcal{M}$,  with affiliation functions $\pi_{h,N} : \mathbb{R} \to \{ 0, \ldots, S- 1\}$ and $ \pi_N = \pi_{h,N} \circ h $, respectively, leading to the empirical quantized observation function 
        \begin{displaymath}
            \bar h_N = \mathbb{E}_{\mu_N}(h \mid \pi_N) = \sum_{i=0}^{S-1} \bar a_{i,N} 1_{M_{i,N}},   
        \end{displaymath}
        where $\bar a_{i,N} = \int_{M_{i,N}} h \, d\mu_N$.
    \item The multiplication operator $T_{\bar h} \in B(L^2(\mu))$ is replaced by the multiplication operator $T_{\bar h_N} \in B(L^2(\mu_N))$. Note that the spectral measure measure $ E_{\bar h_N} $ of the latter satisfies (cf.~\eqref{eqPVMQuant})
        \begin{displaymath}
            E_{\bar h_N}(\{ \bar a_{i,N} \}) = E_{\bar h_N}(\Xi_{i,N}) =  T_{1_{M_{i,N}}}.
        \end{displaymath}
\end{enumerate}

With these definitions, the data-driven formulation of QMDA proceeds entirely analogously to its counterpart from Sections~\ref{secQMDA} and~\ref{secCompactification}. Specifically, selecting a spectral resolution parameter $L \leq N -1 $, and introducing the orthogonal projections $\Pi_{L,N} : L^2(\mu_N) \to L^2(\mu_N)$ mapping into $ \spn\{ \phi_{0,N}, \ldots, \phi_{L-1,N} \} $,  the state $ \rho_{t,N} $ reached at time $ t = q \, \Delta t$ between measurements of $h$, starting from $ \rho_{0,N} \in B(L^2(\mu_N)) $ is given by (cf.~\eqref{eqQMDAApprox}), 
\begin{displaymath}
    \hat \rho_{t,N} = \frac{U^{(q)*}_{L,N} \rho_{0,N} U^{(q)}_{L,N}}{\tr( U^{(q)*}_{L,N} \rho_{0,N} U^{(q)}_{L,N} )}, \quad U^{(q)}_{L,N} = \Pi_{L,N} U^{(q)}_N \Pi_{L,N}.
\end{displaymath}
Moreover, the probability for $ h $ to lie in interval $ \Xi_{i,N} \in \Xi_N$ at a time $ t $ between measurements is determined from (cf.~\eqref{eqPQuantized}), 
\begin{displaymath}
    \hat P_{i,N}(t) = \tr(E_{\bar h_N}(\{\bar a_{i,N} \}) \hat \rho_{t,N}),
\end{displaymath}
while the update from a state  $ \rho^-_N \in B(L^2(\mu_N))$ following a measurement  yielding the value $ a \in \mathbb{R} $, becomes (cf.~\eqref{eqQMDAApprox}) 
\begin{displaymath}
    \hat \rho_{i,N}^+ = \frac{E_{\bar h_N,L}(\{ \bar a_{i,N}\} ) \rho^-_N  E_{ \bar h_N,L}(\{ \bar a_{i,N}\}) }{\tr( E_{\bar h_N,L}(\{ \bar a_{i,N}\} ) \rho^-_N  E_{ \bar h_N,L}(\{ \bar a_{i,N}\}))},
\end{displaymath}
with $ E_{\bar h_N,L}(\Omega) = \Pi_{L,N} E_{\bar h_N,L}(\Omega) \Pi_{L,N} $, $ \forall \Omega \in \mathcal{B}(\mathbb{R})$, and $ i = \pi_{h,N}( a ) $. 

The formulas stated above are sufficient to sequentially perform data assimilation starting from some initial state, which we will set by default to the state $ \bar \rho_N $; see Appendix~\ref{appComputational} for additional details. Then, under the assumptions stated in Sections~\ref{secObs} and~\ref{secSampling}, and an additional mild assumption on the partition $ \Xi $, the data-driven scheme can be shown to converge in the limit of large data, $ N \to \infty $, in the sense that for a fixed spectral resolution $L$ and bounded time interval for data assimilation, the matrix elements of all operators involved, as well as the partition intervals $\Xi_{i,N}$ and assignments $ \pi_N(a) $, converge to their counterparts from Sections~\ref{secQMDA} and~\ref{secCompactification}. This implies, in particular, that all measurement probabilities $ P_{i,N}(t) $ produced by the data-driven assimilation scheme converge. A precise statement of this convergence is made in Theorem~\ref{thmMain}.  It is important to note that the result holds for fixed $L$. Thus, in order to obtain convergence of the data-driven assimilation scheme in a limit of $ N \to \infty$ (training data size) and $L\to\infty$ (spectral resolution), the latter limit must be taken after the former, or, a sequence $N(L)$ with $N\gg L $ must be employed. Effectively, this is because while every matrix element of the form $ \langle \phi_{j,N}, T_N \phi_{k,N} \rangle_{\mu_N} $ converges as $N\to\infty$, where $T_N$ stands here for the shift operator $U^{(q)}_N$ or any of the spectral projectors $E_{\bar h_N}(\{ \bar a_{i,N} \})$, the convergence is not uniform with respect to $j, k$.

\section{\label{secL63}Application to the Lorenz 63 system}

In this section, we apply the data-driven QMDA framework described in Section~\ref{secDataDriven} to data assimilation of the L63  system \cite{Lorenz63} on $M = \mathbb{R}^3$. The L63 system is generated by the smooth vector field $ \vec V : \mathbb{R}^3 \to \mathbb{R}^3$ with components $ (V^1, V^2, V^3 ) $ at $x = (x^1, x^2, x^3) \in M $ given by $V^1 = \sigma(x^2-x^1) $, $V^2 = x^1  (\rho-x^3) -x^2 $, and $V^3 = x^1x^2-\beta x^3$. Here, $\beta$, $\rho$, and $\sigma$ are real parameters set to the standard values  $ \beta = 8/3 $, $\rho = 28$, and $\sigma = 10 $. For this choice of parameters, the L63 system is rigorously known to have a compact attractor $ X \subset M  $  \cite{Tucker99} with fractal dimension $\approx 2.06$ \cite{McGuinness68}, supporting a physical invariant measure $ \mu$ with a single positive Lyapunov exponent $ \Lambda \approx 0.91$ \cite{Sprott03}. Due to dissipative dynamics, the attractor is contained within absorbing balls \cite{LawEtAl14}, playing here the role of the forward-invariant compact set  $\mathcal{X} \subset M$. In light of these facts, all of the assumptions on the dynamical system  made in Sections~\ref{secObs} and~\ref{secObs} rigorously hold. The L63 system is also known to be mixing \cite{LuzzattoEtAl05}, which implies that its associated Koopman unitary group on $L^2(\mu)$ has no nonconstant eigenfunctions. 

In the experiments that follow, we perform data assimilation for the continuous observation function $ h : M \to \mathbb{R} $ projecting onto the first component of the state vector, $h(x) = x^1 $. We consider two training scenarios, namely one where the observation map $F$ is the identity map on $\mathbb{R}^3$ (i.e., the full state vector is observed), and another where only $h$ is observed and an injective map $F$ is built using delays. In both cases, we sample $F$ at $N = \text{64,000}$ points $x_n$, taken along a numerically generated dynamical trajectory at a sampling interval $\Delta t = 0.01 $. The first point $x_0$ in the trajectory is obtained by numerically integrating the L63 system from an arbitrary initial condition in  $\mathbb{R}^3$ for  $N \, \Delta t$ natural time units, and setting $x_0$ to the state reached at the end of that interval. In the experiment with fully observed training data, we set $ F(x_n) = x_n $; the experiments with partial observations use $F(x_n) = h_Q( x_N) $ with $Q = 24$ delays. Using the data $F(x_0), \ldots, F(x_{N-1})$,  we compute orthonormal basis functions $ \phi_{j,N}$ of $L^2(\mu_N)$  as described in Section~\ref{secEig}. Then, using the values $h(x_0), \ldots, h(x_{N-1}) $ of the observation function, we build a partition $ \Xi_N $ with $ S = 32 $ elements and the corresponding projection operators $E_{\bar h_N, L}(\bar a_i)$ as described in Section~\ref{secOpApprox}. Additional details on numerical implementation can be found in Appendix~\ref{appComputational}. 

The experiments with fully observed and partially observed training data use $ L = 1000 $ and 800 basis functions, respectively. In both cases, the time interval between observations during data assimilation is equal to $ 100 \, \Delta t = 1 $, which is comparable to the characteristic Lyapunov timescale $ 1 / \Lambda \approx 1.1$ of the system. In the data assimilation phase, we employ an underlying truth signal $ h(t ) = h(\Phi^t(\tilde x_0)) $ starting from a state $ \tilde x_0 $ sampled on a trajectory independent of the training data. Results from these experiments for the fully and partially observed training data are shown in Figs.~\ref{figPL63} and~\ref{figPL63Q}, respectively.  

\begin{figure}
    \includegraphics[width=\linewidth]{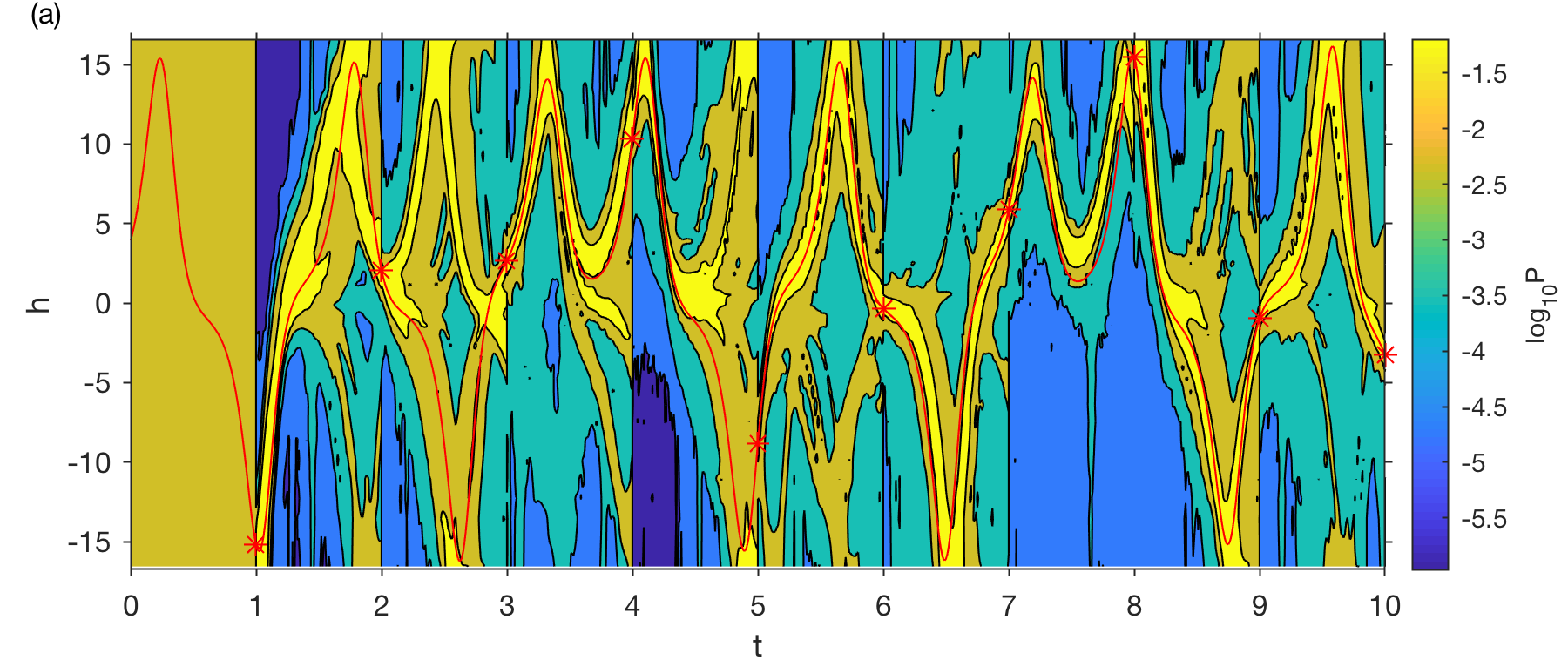}
    \includegraphics[width=\linewidth]{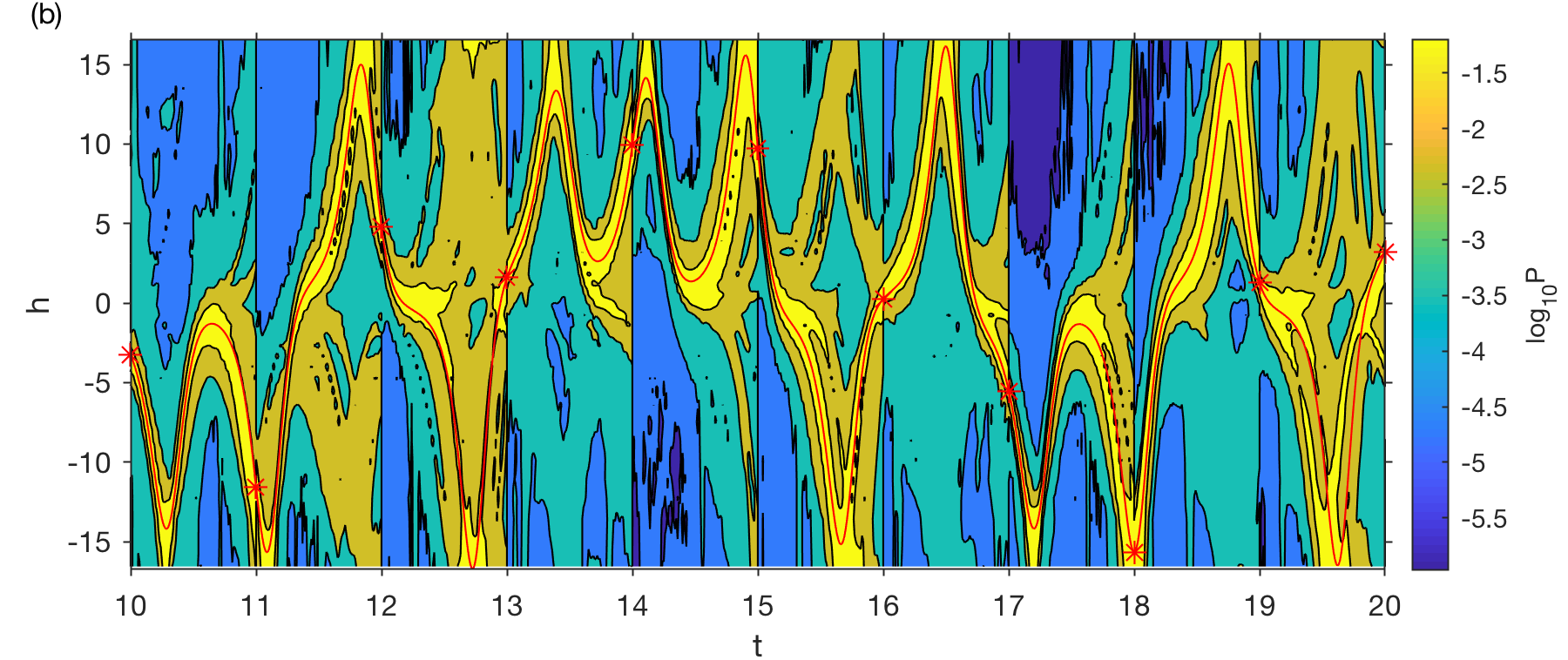}
    \includegraphics[width=\linewidth]{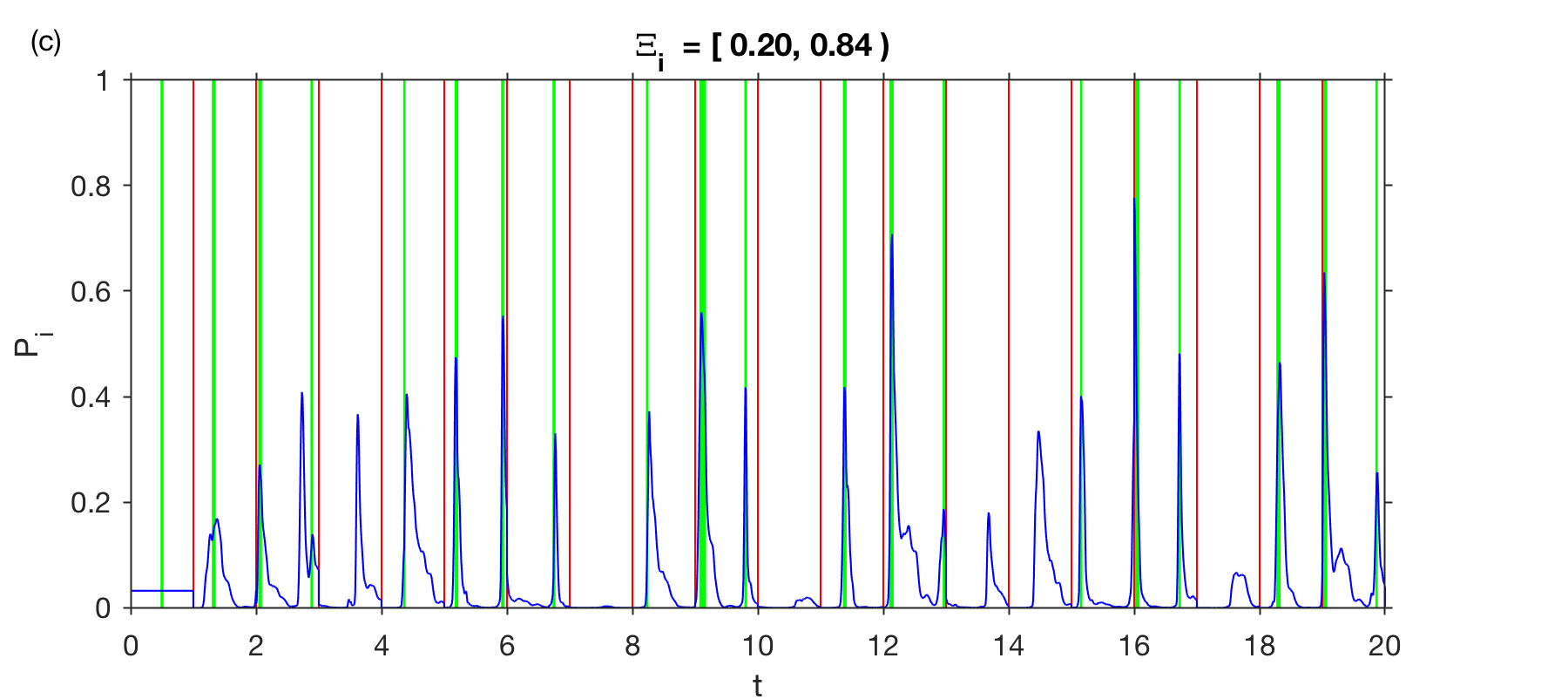}
    \includegraphics[width=\linewidth]{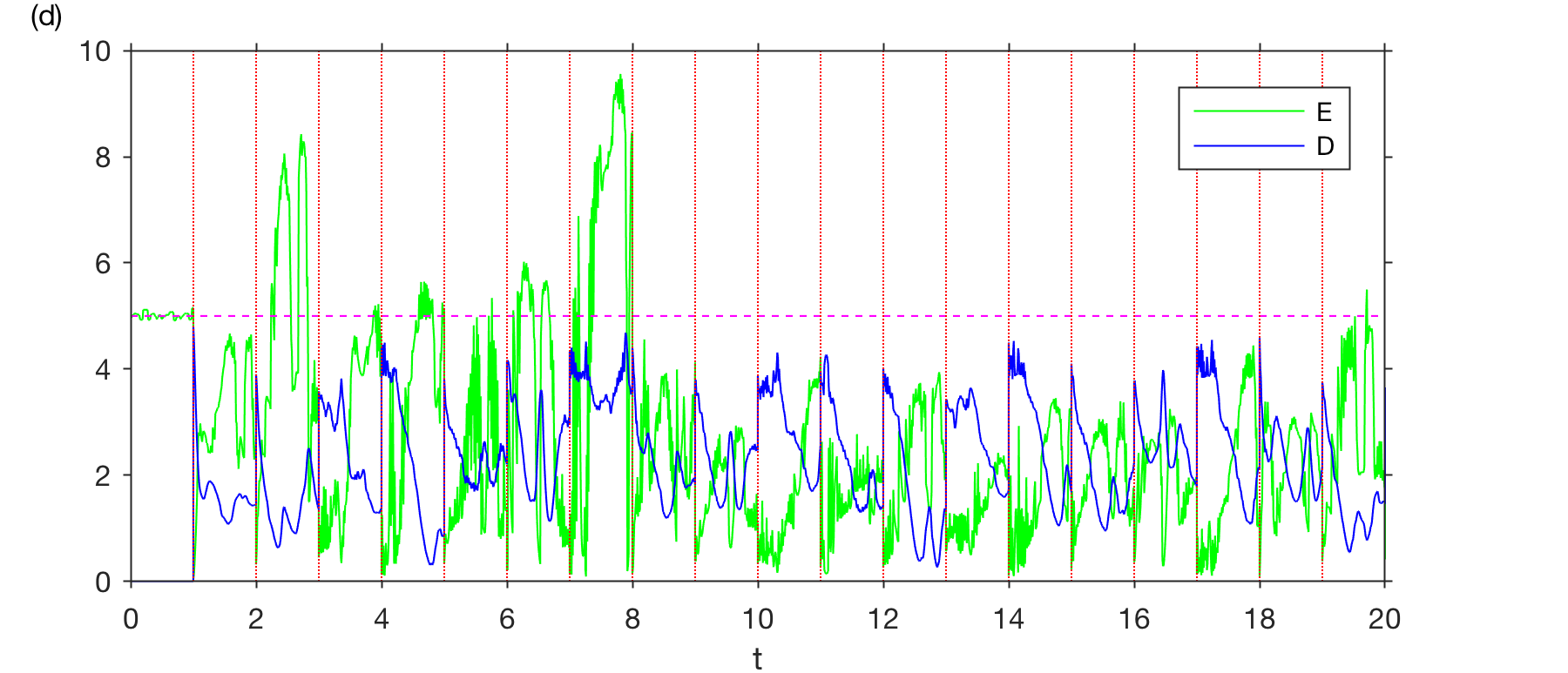}
    \caption{\label{figPL63}Results of QMDA applied to observable $h(x) = x^1$ of the L63 system using the full system state in the training phase (i.e., the eigenfunctions depicted in Fig.~\ref{figL63Phi}). (a, b) Logarithm of the measurement probability $ \hat P_{i,N}(t)$ for $h$ to take values in a partition $ \Xi_N$ of $\mathbb{R}$, containing $ S = 32 $ elements of equal probability mass with respect to the sampling measure $ \mu_N$. Time series $ \hat P_i(t) $ of the probability to obtain a measurement in element $ \Xi_{N,18} \approx [ 0.20, 0.84 ) $ of the partition. Color-coding is as in Fig.~\ref{figPCircleX}(b). (d) Information-theoretic precision and ignorance metrics, $\mathcal{D}(t)$ and $\mathcal{E}(t)$, respectively.}
\end{figure}

\begin{figure}
    \includegraphics[width=\linewidth]{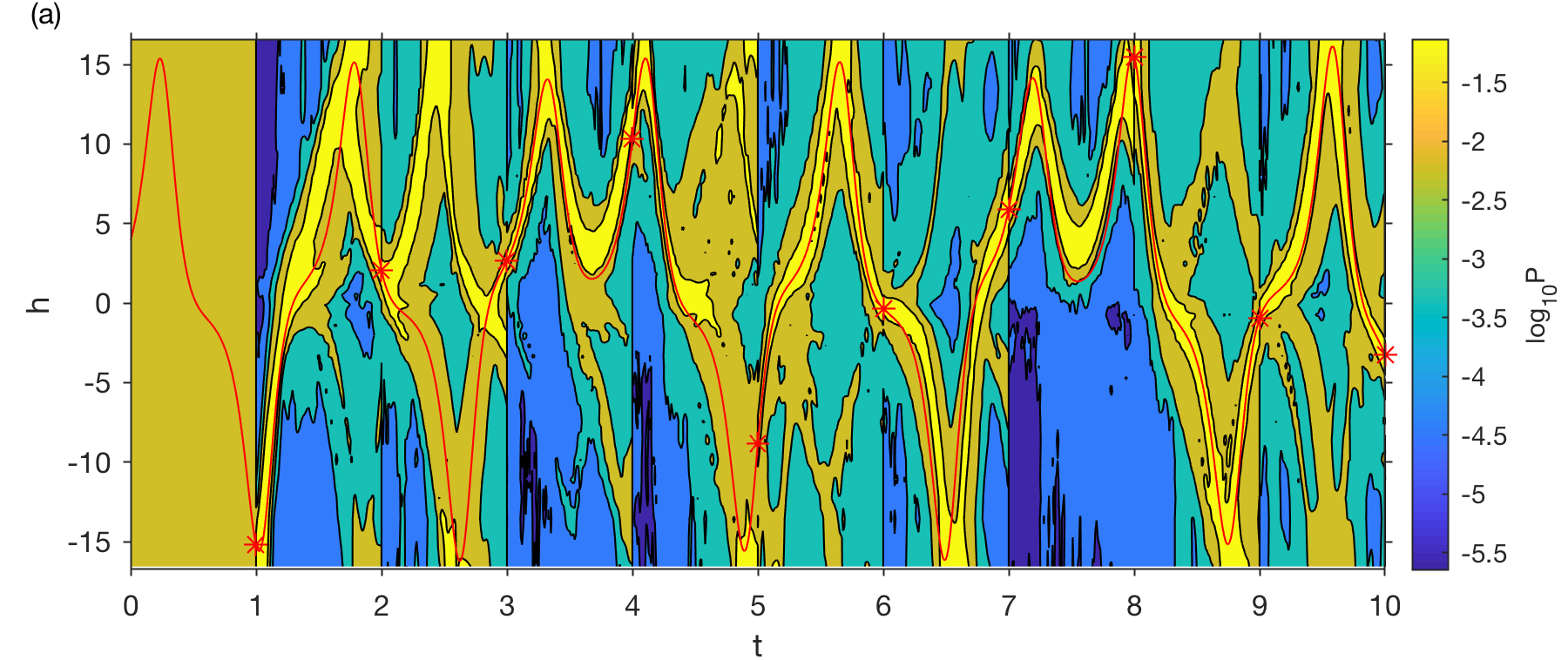}
    \includegraphics[width=\linewidth]{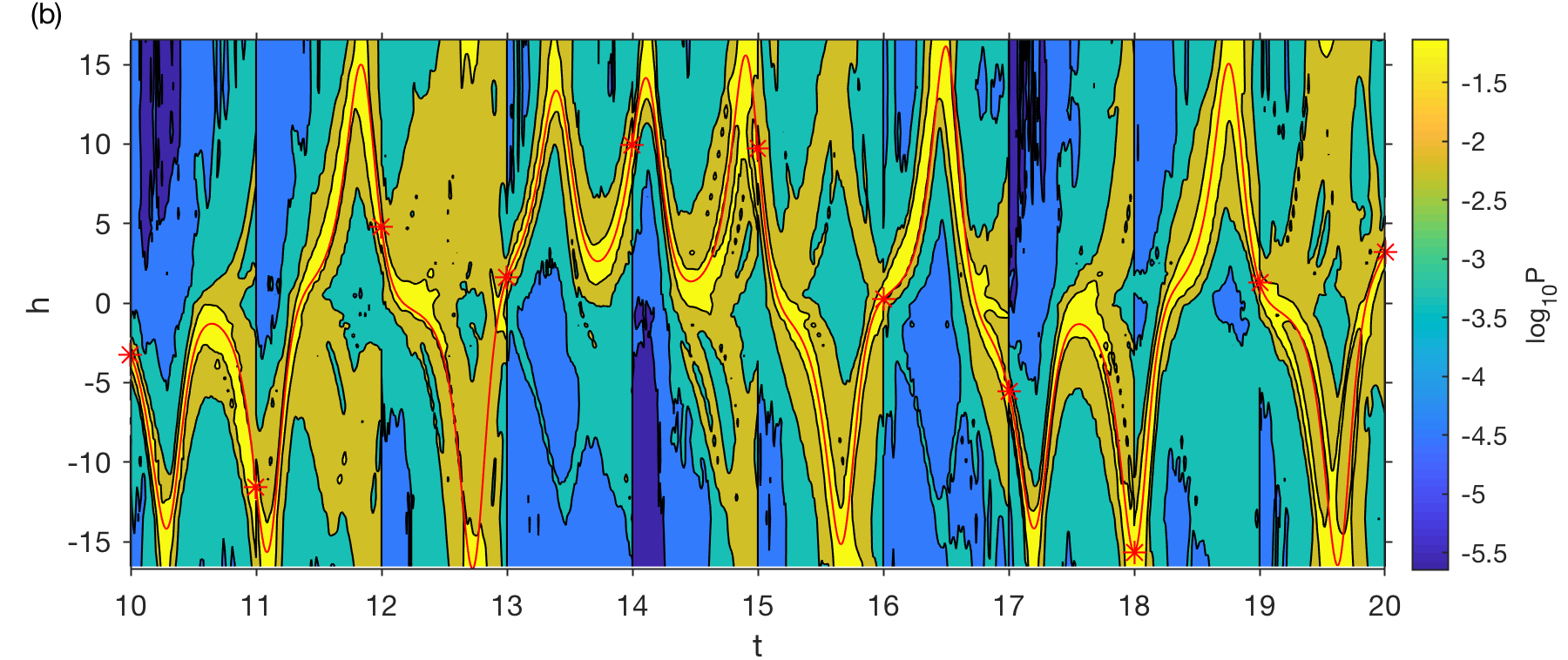}
    \includegraphics[width=\linewidth]{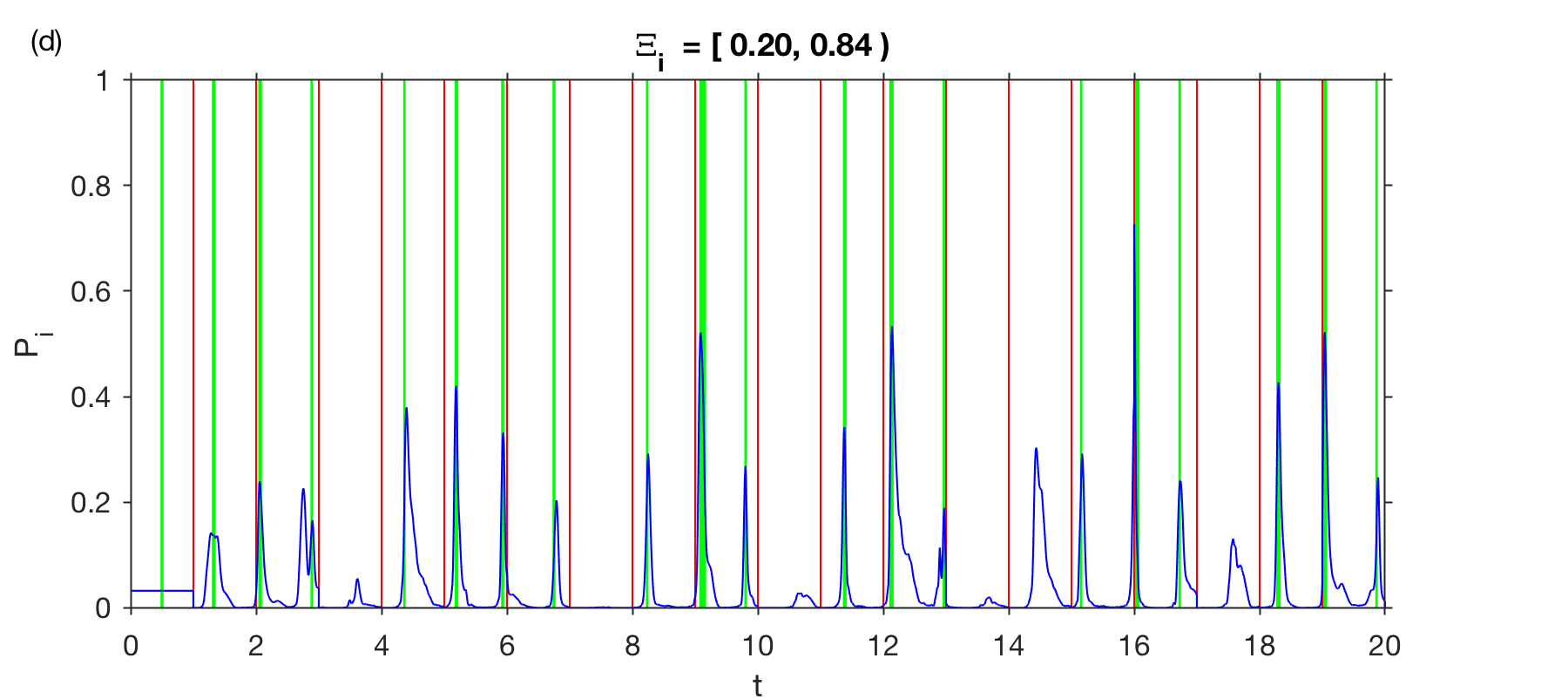}
    \includegraphics[width=\linewidth]{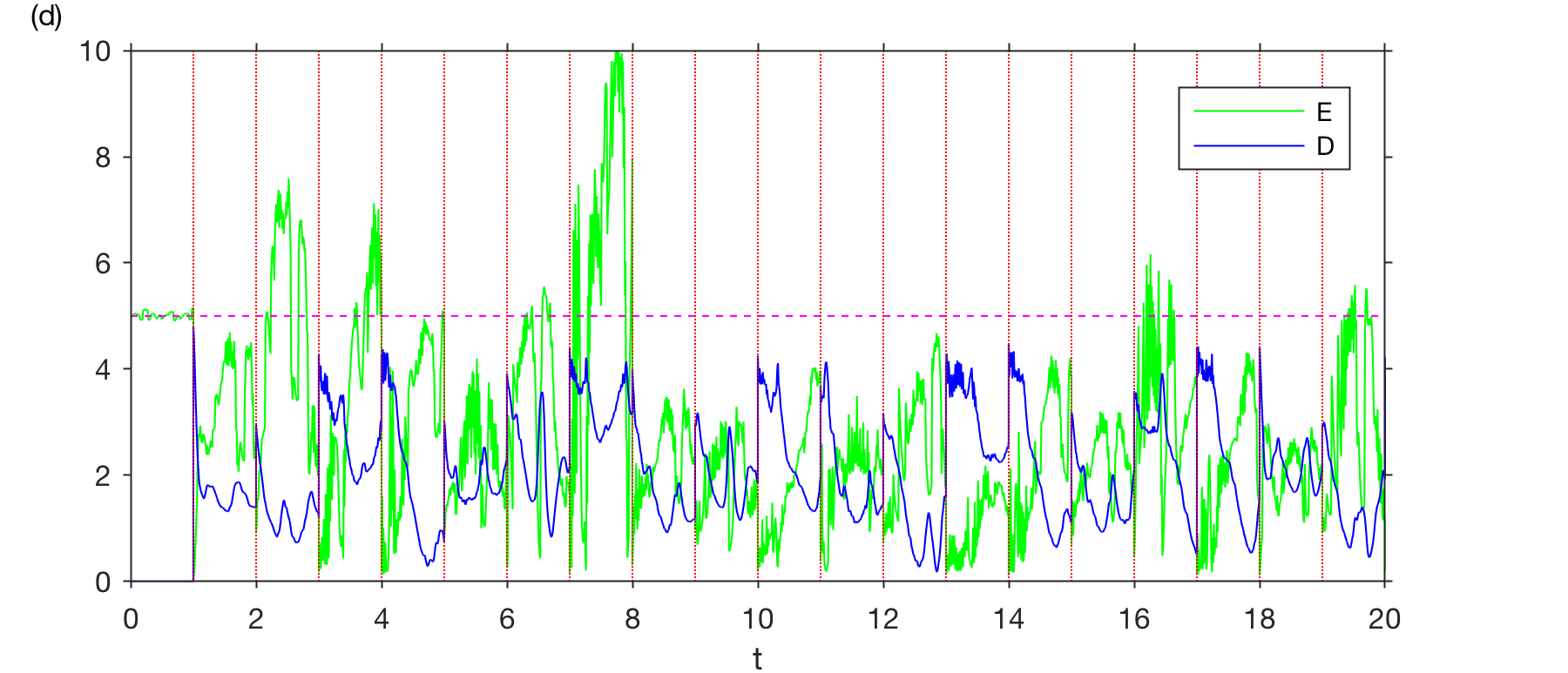}
    \caption{\label{figPL63Q}As in Fig.~\ref{figPL63}, but for partial observations in the training phase, using $ Q = 24 $ delays to compute the basis functions.}
\end{figure}

Starting from the example with the fully observed training data, in Fig.~\ref{figPL63}(a), the initially uniform measurement probability distribution $ \hat P_{i,N}(t) $ associated with the stationary state $\bar \rho_N$ is seen to collapse to a highly sharp probability distribution following a large negative value $h(t)$ measured at $ t = 1 $. After that initial measurement, $ \hat P_{i,N}(t)$ tracks the evolution of $h(t)$ fairly well, though with increasing uncertainty and development of some bimodality for $t \gtrsim 1.2$. The second measurement at $ t= 2 $ produces a value $h(t)$ significantly closer to the origin---this is presumably less informative than the $ t = 1$ measurement since $h(t) \approx 0$ corresponds to the mixing region between the two lobes of the L63 attractor. The lack of information in the $t=2$ measurement is manifested in the ensuing evolution of $ \hat P_{i,N}(t)$, which exhibits significant bimodality and erroneously places the highest probability on positive values of $h$, whereas the true signal takes negative values. This error is clearly visible in the $ \mathcal{E}(t)$ metric in Fig.~\ref{figPL63}(d), which exhibits a pronounced increase to greater than $ \log_2 S = 5$ values over the time interval $ ( 2, 3 ) $. In spite of the poor data assimilation performance for $ t \in ( 2, 3)$, when the next measurement $ h(t_3) $ comes in at $ t = 3 $, the ensuing measurement probability $ \hat P_{i,N}(t)$ tracks the true signal with significantly higher accuracy, despite the fact that $h(t_3) $ is comparably close to zero as $h(t_2)$.  This improvement of skill demonstrates that the data assimilation state $\hat \rho_{t,N}$ can progressively become more informative from a succession of uninformative measurements. Indeed, as shown in Fig.~\ref{figPL63}(d), following a spike in $\mathcal{E}(t)$ for $t \in ( 7, 8 )$, the data assimilation system appears to settle in a regime where $\mathcal{E}(t)$ is either significantly smaller than $5 $, or slightly exceeds that threshold (e.g., the interval $t \in (19,20)$). In general, these periods of larger error $\mathcal{E}(t)$ appear to correlate with observations $h(t)$ close to zero. For instance, see the measurement at $t = 19 $ in Fig.~\ref{figPL63}(b), which is followed by probability distributions $ \hat P_{i,N}(t)$ of comparatively large uncertainty. It is also worthwhile noting that the precision metric $\mathcal{D}(t)$ in Fig.~\ref{figPL63}(d) exhibits markedly more appreciable drops between measurements than in the case of the circle rotation in Fig.~\ref{figPCircleX}(c, d), as expected from the mixing nature of the L63 dynamics.    

Turning now to the example with partial observations in the training phase, a comparison between Figs.~\ref{figPL63} and~\ref{figPL63Q} shows a broadly consistent behavior with the experiment trained with full observations.  That is, following an initial period $ t \in [ 0, 8 )$ which exhibits similar errors to the fully observed case, the data assimilation system reaches a regime of smaller $ \mathcal{E}(t) $ metric, characterized by moderate and infrequent crossings of the $\mathcal{E}(t) = 5$ threshold when $h(t)$ takes values close to zero (e.g., $t = 16$ and 19). Overall, this behavior demonstrates that the delay-coordinate mapping was able to successively recover lost information due to partial observations in the training phase, enabling the construction of a purely data-driven data assimilation scheme for this chaotic dynamical system.

\section{\label{secDiscussion}Discussion}

To place the QMDA approach proposed in this paper in context, we now discuss some of its advantages and shortcomings compared to ``classical'' sequential data assimilation schemes. Here, by ``classical'' we mean data assimilation approaches whose ultimate goal is to perform Bayesian inference; that is, compute the Bayesian posterior distribution of a quantity of interest (which may be the full system state), given a history of observations made on the system \cite{LawStuart12}. In practical applications involving complex systems, rigorous Bayesian inference is not feasible for a variety of reasons, including imperfect or computationally intractable equations of motion for the system dynamics, unknown observational modalities, and singular probability measures (particularly in the setting of deterministic dynamics studied in this paper). As a result, starting from the original work of Kalman \cite{Kalman60}, a vast array of approximation techniques has been developed and currently deployed in operational environments \cite{LawEtAl15,MajdaHarlim12,Kalnay03,LahozEtAl10,ThrunEtAl05}. An attractive feature of QMDA is that it naturally circumvents a number of the challenges in classical data assimilation, and thus avoids the need for ad hoc approximations, by employing intrinsic linear operators to represent the state, dynamics, and observables. 

First, the density operator $\rho$ representing the state of the data assimilation system (i.e., the analog of the posterior measure in Bayesian data assimilation; see Table~\ref{tableComparison}) is a well-behaved linear operator on the $L^2$ space associated with an invariant measure $\mu$ of the system, even if the support of $\mu$ is a null set with respect to an ambient measure on state space (e.g., Lebesgue measure), and/or does not have a smooth structure. This situation clearly occurs in the L63 example of Section~\ref{secL63}, where $\mu$ is supported on the fractal Lorenz attractor, but also in systems with considerably simpler dynamics. For instance, the circle, $S^1$, employed as the state space of the periodic dynamical system in Section~\ref{secCircle}, can be thought of as the support of an invariant measure of the simple harmonic oscillator on $\mathbb{R}^2$ corresponding to constant energy, and treating $\mathbb{R}^2$ as the ambient state space equipped with the Lebesgue (area) measure, makes $S^1$ a zero-measure set. In a Bayesian data assimilation setting, this means that it is not possible to represent the posterior measure by a density function, necessitating in practice some form of approximation, including addition of stochastic noise to regularize the dynamics. 

Particle filters \cite{VanLeeuwenEtAl19} perform this approximation by representing the posterior through a weighted ensemble of Dirac measures (the particles), and can theoretically converge to the true posterior in a large-ensemble limit. In practice, however, these methods suffer from well-known issues of ensemble collapse \cite{ChorinMorzfeld13}, particularly in high-dimensions and/or in the presence of dissipation. As a result, one must resort to some type of ensemble regeneration procedure, with generally difficult to control convergence guarantees. Methods that are not based on sampling frequently invoke Gaussianity assumptions, leading to popular schemes such as the 3DVAR filter, the extended Kalman filter, and the ensemble Kalman filter \cite{LawEtAl15,MajdaHarlim12,Kalnay03}. Despite their popularity, theoretical studies on the behavior of these methods have been limited to particular cases, and have mainly focused on filter accuracy, as opposed to convergence to the full Bayesian posterior distribution. See, e.g., Ref.~\cite{LawEtAl14} for an analysis of the 3DVAR filter applied to linear observations of the L63 system. In contrast, the consistency of the data-driven formulation of QMDA in the large data limit, i.e., its ability to converge to the ``true'' quantum mechanical state update  (Step~DA5 in Section~\ref{secQMDA}), is essentially a direct consequence of the approximability of trace-class operators by finite-rank operators (e.g., Theorem~\ref{thmMain}). While filter accuracy results analogous to those in Ref.~\cite{LawEtAl14} lie outside the scope of this work, it is expected that the framework of linear operator theory on Hilbert spaces could be used to address such questions in a unified manner for broad classes of systems.

Next, with regards to the representation of the forward dynamics, the Koopman operator formalism employed by QMDA is again intrinsically linear, and as discussed in Section~\ref{secDataDriven} is amenable to data-driven approximation through the use of kernel and delay-coordinate techniques without requiring prior knowledge of the equations of motion and/or diffusion regularization, while obeying rigorous convergence guarantees. Previously, delay-coordinate maps  and kernel methods  have been employed in data-driven filtering algorithms to reconstruct unknown dynamics \cite{HamiltonEtAl16}, and correct for observational biases \cite{BerryHarlim17}, respectively. These methods are, however, closer to classical data assimilation approaches since their focus is on approximating the Bayesian posterior distribution as consistently as possible. Other data-driven approaches to filtering have employed neural network architectures in either of the forecast \cite{OualaEtAl18} or analysis steps \cite{CintraEtAl18}. 

Of course, it should be kept in mind that an operator-theoretic, fully empirical approximation of the dynamics does not come without its disadvantages. In particular, in many applications of interest one \emph{does} have access to a first-principles parametric model (e.g., a numerical weather model), which even if imperfect, may be indispensable in intrinsically high-dimensional applications.  At present, we do not have a technique allowing us to seamlessly combine a data-driven Koopman operator model with a first-principles state space model, although recent techniques on semiparametric modeling \cite{BerryHarlim16b} and the Mori-Zwanzig formalism \cite{GouasmiEtAl17} could pave the way for such developments. That being said, it should be noted that in a number of phenomena of interest (e.g., large-scale coherent patterns in climate dynamics such as the El Ni\~no Southern Oscillation and the Madden-Julian Oscillation) there are simply no ``perfect'' first-principles governing equations, while the effective dimension of the dynamics is moderate. In such scenarios, fully data-driven filtering approaches such as QMDA may be competitive, or even exceed the performance of large-scale parametric data assimilation systems.         

As a final remark, we note that the quantum mechanical representation of observables through intrinsically linear multiplication operators, in conjunction with the spectral discretization approach described in Section~\ref{secCompactification}, allows QMDA to naturally handle nonlinear observation functions $h$, as well observational noise. In particular, even though we did not study this topic explicitly here, the number of elements $S$ in the partition $ \Xi $ employed for spectral discretization could be selected according to a desired tolerance to noise. That is, in general, the smaller $S$ is, the more robust the state update step DA5$'$ becomes (see Section~\ref{secCondAv}), at the expense of a loss of resolution afforded by the measurements. Moreover, if prior knowledge about the noise statistics is available, the elements of $\Xi$ could be chosen non-uniformly so as to ensure high robustness in subsets of the range of $h$ where the noise has high strength, and high resolution in subsets where the noise is weaker. It is also worthwhile noting that the prediction output of QMDA for observables is intrinsically probabilistic; that is, according to step DA4, the density operator and spectral measure of an observable provide the probability for it to take values in arbitrary Borel sets. This output can be further post-processed to yield the mean, variance, skewness, and other statistical quantities of interest. In contrast, classical data assimilation techniques utilizing Gaussian approximations only dynamically evolve the mean and covariance.

\section{\label{secConclusions}Conclusions}

In this work, we have developed a framework for sequential data assimilation in measure-preserving ergodic dynamical systems combining elements of operator-theoretic ergodic theory and quantum mechanics. A key aspect of this approach has been to transcribe the Dirac--von Neumann axioms of quantum dynamics and measurement to the setting of measure-preserving ergodic dynamics by choosing as the quantum mechanical Hilbert space the $L^2(\mu)$ space associated with the invariant measure of the dynamics, and as the Heisenberg evolution operators the unitary Koopman operators acting on $L^2(\mu)$. Also in direct analogy with quantum mechanics, we represent the time-dependent state of the data assimilation system by a density operator on $L^2(\mu)$ and the system observation function $h$ by its corresponding self-adjoint multiplication operator $ T_h$. With these identifications, a quantum mechanical data assimilation (QMDA) scheme follows very naturally by allowing the state to evolve under unitary dynamics induced by Koopman operators between measurements and projective dynamics under the spectral projectors of the observation operator. 

One issue that such a scheme must confront is that the multiplication operators associated with typical observation functions will have continuous spectrum. Here, we addressed this issue via a quantization approach, whereby $h $ is replaced by a discrete variable $\bar h$ such that the corresponding multiplication operator $ T_{\bar h} $ has pure point spectrum. In particular, $ \bar h $ was constructed by binning $h$ into bins of equal probability mass with respect to $ \mu$, but one could employ different averaging approaches, e.g., to take into account observational noise. We also studied the problem of constructing a data-driven formulation from a finite collection of time-ordered observations of the system state or $h$, assuming no prior knowledge of the equations of motion. This formulation employed operator approximation techniques  in a basis of $L^2(\mu)$ learned from the observed data  \cite{BerryEtAl15,Giannakis17,DasEtAl18}, leading to representations of the Koopman and measurement operators via matrices that provably converge in an asymptotic limit of large data (Theorem~\ref{thmMain} in Appendix~\ref{appConvergence}).       

An attractive feature of the QMDA approach presented here is that it requires no ad hoc approximations of the dynamics and/or observation modality, which are frequently necessary in order to apply classical data assimilation techniques to measure-preserving deterministic systems. Indeed, as we demonstrated here with examples,  whether the underlying dynamics is a periodic rotation on a circle (Section~\ref{secCompactification}), or a mixing system with a fractal attractor (Section~\ref{secL63}), makes little difference from a methodological standpoint in the context of QMDA. In both cases, we saw that the method can successfully capture highly non-Gaussian features of the measurement distribution that accurately track the evolution of the assimilated observable. 

Another advantageous aspect of QMDA is that it outputs full probability distributions, as opposed to point forecasts such as mean or maximum likelihood estimates. This output can be post-processed in a variety of ways to enable uncertainty quantification, as well as probabilistic decision-making in an operational environment. We also saw that the probability distribution outputs of QMDA lead to natural information-theoretic metrics for quantification of the precision and ignorance of data assimilation. Such metrics have been shown to provide more informative model assessment and validation than conventional root mean square error and pattern correlation metrics in a different context \cite{MajdaQi18}.

Before closing, we outline a few aspects of QMDA that warrant future study and potential improvement, some of which have been already alluded to in Section~\ref{secDiscussion}. First, while in this paper we have shown that the method converges in a limit of large data, one aspect of convergence that has not been addressed is convergence under refinement of the partition employed for spectral discretization of $T_h$. It is possible that a general treatment of this problem in the case of observables with continuous spectrum would employ a rigged Hilbert space structure \cite{BohmGadella89}, allowing $ \rho $ to be extended to an operator on distributions. Second, we have restricted ourselves to the case of scalar-valued observation maps. An interesting question would be how to carry out an analogous QMDA construction for vector-valued functions, possibly taking values in an infinite-dimensional Banach space. Such a construction may have connections with the framework of quantum field theory. Algorithmically, it could be implemented by replacing the scalar-valued kernels employed here in the construction of the data-driven basis by operator-valued kernels appropriate for spaces of vector-valued functions \cite{SlawinskaEtAl18}.  Finally, an important task would be to devise ways of effectively coupling an observable-centric scheme such as QMDA, which employs linear operators on function spaces at its core, with dynamical models based on discretizations of the dynamics in state space (e.g., a partial differential equation model governing fluid flow). With the advent of quantum information processing technologies, it is possible that schemes such as QMDA could provide guidance to the design of next-generation dynamical models of complex systems.  

\begin{acknowledgements}
    The author acknowledges support by ONR YIP grant N00014-16-1-2649 and NSF grant 1842538. He is also grateful to the Department of Computing and Mathematical Sciences at the California Institute of Technology and in particular his host, Andrew M.\ Stuart, for their hospitality and for providing a stimulating environment during a sabbatical, when the majority of this work was completed. 
\end{acknowledgements}

\appendix

\section{\label{appConvergence}Convergence in the limit of large data}

In this appendix, we state and prove the following theorem establishing asymptotic consistency of the data-driven QMDA scheme from Section~\ref{secDataDriven} in the limit of large data. In what follows, we will say that a sequence of $ S$-element partitions  $ \Xi_N$ of $\mathbb{R}$ converges as $N \to \infty$ to an $S$-element partition $\Xi$ if all boundary points of the elements of $ \Xi_N $ (ordered in increasing order at each $N$) converge to the corresponding boundary points of $ \Xi $.   

\begin{thm} \label{thmMain}Consider data assimilation with a bounded observation function $ h \in L^\infty(\mu) $ via the scheme of Section~\ref{secDataDriven}, using a partition $ \Xi_N = \{ \Xi_{0,N}, \ldots, \Xi_{S-1,N} \} $ of $\mathbb{R}$ determined through the empirical quantile function $ \cdf^{-1}_{h,N} $, and starting from the stationary state $ \bar \rho_N$. Assume that the partition $ \Xi = \{ \Xi_0, \ldots, \Xi_{S-1} \} $ of $\mathbb{R} $ associated with the true quantile function $ \cdf^{-1}_h $ is such that $ \mu_h( \{ \xi_k \} ) = 0 $ for all boundary points $ \xi_k $ of the $ \Xi_j \in \Xi $. Then, under the assumptions of Sections~\ref{secObs} and~\ref{secSampling}, for any spectral resolution parameter $ L $, the partitions $ \Xi_N $ and corresponding measurement probabilities $ \hat P_{i,N}(t_j) $ at time $ t_j = j \, \Delta t $, $ j \in \mathbb{N}_0 $, converge as $ N \to \infty $ to $ \Xi $ and the probabilities $ \hat P_{i}(t_j) $ obtained via the scheme of Section~\ref{secCompactification}, respectively, using the same spectral resolution parameter $L$.      
\end{thm}

\begin{proof}

    It suffices to show that, as $N \to \infty$, (i) the elements of all $L \times L $ matrices representing the operators employed in the data-driven scheme converge to their counterparts from Section~\ref{secCompactification}; and (ii) $ \Xi_N $ converges to $ \Xi $.  Note, in particular, that the latter convergence implies that the affiliation functions $ \pi_{h,N} $ converge to $ \pi_N $ pointwise, and thus that the finitely many evaluations of $ \pi_{h,N} $ during the time interval $ [ 0, t_j ] $ also converge.  

    Starting from (ii), recall that the boundary points $ \xi_k $ and $ \xi_{k,N} $ of the intervals in $ \Xi $ and $ \Xi_N $, respectively, are obtained by evaluating the corresponding quantile and empirical quantile functions at the same quantile points $ b_k \in ( 0, 1 ) $; that is, $ \xi_k = \cdf_h^{-1}(b_k) $ and $\xi_{k,N} = \cdf_{h,N}^{-1}(b_k)$. Because $ \mu_h( \{ \xi_k \} ) = 0 $,  $ \xi_k $ and $ b_k$ are continuity points of $ \cdf_h $ and $ \cdf_h^{-1} $, respectively. As a result, by the assumed weak convergence of the sampling measures $ \mu_N $ to $ \mu $ (which implies weak convergence of the corresponding pushforward measures $\mu_{h,N}$ under $h$ to $ \mu_h$), the values $ \cdf^{-1}_{h,N}( b_k) $ of the empirical quantile functions converge, as $N \to \infty$, to $ \cdf^{-1}_h(b_k) $ \cite{FristedtGray97}. This shows that $ \Xi_N $ converges to $ \Xi $.

    Turning to (i), the operators that we need to consider are (a) the initial state $ \bar \rho_N$; (b) the shift operator $U^{(q)}_{N}$; and (c) the projection operators $E_{\bar h_N}(\{ \bar a_{i,N} \}) = T_{1_{M_{i,N}}} $. Indeed:

    \begin{enumerate}[(a), wide]
        \item The matrix elements $ \langle \phi_{j,N}, \bar \rho_N \phi_{k,N} \rangle_{\mu_N} = \delta_{j0}\delta_{k0} $ are trivially equal to $ \langle \phi_{j}, \bar \rho \phi_{k} \rangle_{\mu} = \delta_{j0}\delta_{k0} $.
        \item Because the basis functions $ \varphi_{j,N}$ converge uniformly to $ \varphi_j $ on $\mathcal{X}$ (see Section~\ref{secEig}), the matrix elements  of the shift operator $U^{(q)}_N $  converge to those of the Koopman operator $U^t$ at $ t = q \, \Delta t $, viz. 
            \begin{multline*}
                \lim_{N\to\infty}\langle \phi_{j,N}, U^{(q)}_{N} \phi_k \rangle_{\mu_N} \\
                \begin{aligned}
                    &= \lim_{N\to\infty} \frac{1}{N} \sum_{n=0}^{N-1} \phi_{j,N}(x_n) \phi_{k,N}(x_{n+q}) \\
                    &= \lim_{N\to\infty} \frac{1}{N} \sum_{n=0}^{N-1} \varphi_{j,N}(x_n) \varphi_{k,N}(x_{n+q}) \\
                &= \lim_{N\to\infty} \frac{1}{N} \sum_{n=0}^{N-1} \varphi_{j,N}(x_n) ( \varphi_{k,N} \circ \Phi^{t} )(x_{n}) \\
                &= \int_M \varphi_{j}(x) ( \varphi_{k} \circ \Phi^{t} )(x)  \, d\mu(x)\\
                &= \int_M \phi_{j}(x) U^t \phi_k(x)  \, d\mu(x)\\
                &= \langle \phi_j, U^t \phi_k \rangle_\mu.
            \end{aligned}
        \end{multline*}
    \item Let $ \Xi_{i,N} $ and $ \Xi_i $ be the $i$-th elements of $ \Xi_N $ and $\Xi$, respectively, where $ i \in \{ 0, \ldots, S-1 \} $ is arbitrary. We must show that, as $N \to \infty$,  $ \langle \phi_{j,N}, T_{1_{M_{i,N}}} \phi_{k,N} \rangle_{\mu_N} $ converges to $ \langle \phi_j, T_{1_{M_i}} \phi_k \rangle_\mu $. To that end, observe that  $ \langle \phi_{j,N}, T_{1_{M_{i,N}}} \phi_{k,N} \rangle_{\mu_N} = \tilde \mu_N(M_{i,N}) $ and  $ \langle \phi_j, T_{1_{M_i}} \phi_k \rangle_\mu = \tilde \mu(M_i) $, where $\tilde \mu_N $ and $ \tilde \mu $ are finite, signed Borel measures on $M$ such that $ \tilde \mu_N(\Omega) = \int_\Omega \varphi_{j,N} \varphi_{k,N} \, d\mu_N $ and $ \tilde \mu(\Omega) = \int_\Omega \varphi_{j} \varphi_{k} \, d\mu $. As a result, the claim will follow if it can be shown that
        \begin{align}
            \nonumber \lvert \tilde \mu_N(M_{i,N}) - \tilde \mu(M_i) \rvert &\leq \lvert \tilde \mu_N(M_{i,N}) - \tilde \mu_N(M_i) \rvert \\
            \label{eqMuBound}& \quad + \lvert \tilde \mu_N(M_i) - \tilde \mu(M_i) \rvert
        \end{align}
        vanishes as $N \to \infty$. Now, it is straightforward to verify that, by uniform convergence of $ \varphi_{j,N}$ to $ \varphi_j$, $\tilde \mu_N $ converges weakly to $ \tilde \mu $. As a result, because $ M_i $ is a continuity set of $ \tilde \mu $ (i.e., $ \tilde \mu( \partial M_i ) = 0 $, which follows from the fact that $ \mu_h(\{\xi_k\})=0 $ for all boundary points $ \xi_k$), $ \tilde \mu_N( M_i) $ converges to $ \tilde \mu(M_i) $, and the second term in the right-hand side of~\eqref{eqMuBound} vanishes. Similarly, it follows by uniform convergence of $ \varphi_{j,N} $ to $ \varphi_{j}$ that there exists a constant $C $ such that $\lvert \tilde \mu_N(M_{i,N}) - \tilde \mu_N(M_i) \rvert \leq C \lvert \mu_N(M_{i,N}) - \mu_N(M_i) \rvert$. Therefore, because $ \mu_N(M_{i,N}) = 1/ S $ by construction, we can conclude that the first-term in the right-hand side of~\eqref{eqMuBound} will also converge to zero if it can be shown that $\mu_N(M_i)$ converges to $1/S$. The latter follows immediately from the weak convergence of $ \mu_N $ to $ \mu $ and  fact that $M_i $ is a continuity set of $\mu$. 
    \end{enumerate}
    This completes the proof of Claim (i) and of the theorem.
\end{proof}

It should be noted that the assumption in Theorem~\ref{thmMain} that the boundary points have vanishing $ \mu_h $ measure is mild, in the sense that if not satisfied, the condition can be met by shifting the problematic $ \xi_k $ by arbitrarily small amounts.


\section{\label{appComputational}Computational considerations}

In this appendix, we outline aspects of the numerical implementation and computational cost of the data-driven implementation of the QMDA framework described in Section~\ref{secDataDriven}, and employed in the numerical experiments of Section~\ref{secL63}. Algorithmically, the main steps of the procedure are (i) computation of the eigenvectors $ \vec \phi_{j,N}$ representing the data-driven basis elements $\phi_{j,N}$ from the training data; (ii) representation of the Koopman operator (approximated by the shift operator) and spectral projectors in this basis by matrices; and (iii) execution of the prediction-correction data assimilation cycle from sequential observations of the system. A key element of this procedure is that following an expensive, offline calculation step to compute the $\phi_{j,N}$, the cost of operator representation is controlled by the spectral resolution parameter $L$, which is independent of the dimension of the ambient data space and number of training samples, thus aiding the scalability of the framework to large training datasets. The numerical experiments in Section~\ref{secL63} were carried out using a Matlab code for QMDA running on a desktop-class workstation of modest specifications at the time of writing of this paper (Intel(R) Core(TM) i7-3770 CPU at 3.40 GHz, with 32GB of memory).

\subsection{\label{appBasis}Data-driven basis}

The computation of the $\vec\phi_{j,N}$ proceeds via well-established kernel algorithms for machine learning. In this step, a major component of the computational cost, both in terms of CPU time and memory, is associated with the computation of the $N \times N$ kernel matrix $\bm G$ associated with the observations $F(x_n) \in \mathbb{R}^m$. Here, we compute this matrix in brute force, resulting in an $O(mN^2)$ time cost, but the calculation is trivially parallelizable. As is customary, to address the memory cost for $ \bm G$, which is nominally $O(N^2)$, we take advantage of the exponential decay of the kernel in~\eqref{eqKVB}, and approximate $\bm G$ by a sparse, symmetric  $N\times N$ matrix $ \hat{\bm G} $, such that $ \hat G_{mn} = G_{mn} $ if data point $F(x_m) $ is in the $r$-th nearest neighborhood of $F(x_n)$, of $F(x_n)$ is in the $r$-th nearest neighborhood of $F(x_m)$ for some neighborhood parameter $ r \ll N $, and $ \hat G_{mn} = 0 $ otherwise. In particular, in the experiments of Section~\ref{secL63} we use $ r = 5000 $, corresponding to $ \simeq 8\%$ of the $N = \text{64,000} $ training data points. 

We compute leading $L$ eigenvectors $ \vec \phi_0, \ldots, \vec \phi_{L-1}$ of $ \hat{\bm G} $  using Matlab's \texttt{eigs} solver, which is based on implicitly restarted Arnoldi methods in the ARPACK library \cite{LehoucqEtAl98}. The eigenvectors $ \vec \phi_j $ then provide representations of the basis elements $ \phi_{j,N} $ (see Section~\ref{secEig}).  Elsewhere, we have demonstrated the feasibility of this implementation for computing eigenfunctions from high-dimensional datasets of moderate sample number, $(d,N) = O(10^6,10^4)$ \cite{GiannakisEtAl18b}, or datasets of moderate dimension and high sample number, $(d,N) = O(10^2,10^6)$ \cite{GiannakisEtAl19}. In the latter case, it should be possible to speed up the kernel matrix calculation using tree-based \cite{AryaEtAl98} or randomized \cite{JonesEtAl11} approximate nearest-neighbor algorithms, though we have not explored such options in the present work. 

\subsection{\label{appOps}Operator representation}

Having obtained the data-driven basis functions $\phi_{j,N}$, we proceed to construct the $L \times L $ matrices representing the Koopman operators and spectral projectors from Section~\ref{secOpApprox}. 

In the case of the Koopman operators we represent $ U^{(q)}_{L,N} $ for each time step $ q \in \mathbb{N} $ of interest by a matrix $\bm U^{(q)} $ with elements 
    \begin{align*}
        \bm U^{(q)}_{jk} &= \langle \phi_{j,N}, U^{(q)}_{L,N} \phi_{k,N} \rangle_{\mu_N} = \frac{1}{N} \sum_{n=0}^{N-1-q} \vec \phi_{j,n} \vec \phi_{k,n+q},
    \end{align*}
where $0 \leq j,k \leq L - 1$, and $ \vec \phi_{j,n}$ denotes the $n$-th component of $ \vec \phi_j $. The computation cost to form this matrix is $O(N L^2)$. In order to carry out the forward evolution of the density operator between measurements (step (DA2) in Section~\ref{secQMDA}) one requires the formation of $\bm U^{(q)}$ at least for $q$ equal to number of timesteps $ \Delta t$ in each data assimilation interval (e.g., in Section~\ref{secL63}, $q=100$). The matrices $ \bm U^{(q)} $ can be computed for other values of $q  $ if forecast output at other times is desired. In particular, to obtain the results in Figs.~\ref{figPL63} and~\ref{figPL63Q} we employ $\bm U^{(0)}, \ldots, \bm U^{(100)}$. Alternatively, one can compute the 1-step matrix $\bm U^{(1)}$, and use the matrix power $(\bm U^{(1)})^q $ instead of $ \bm U^{(q)} $. This approach avoids the storage cost for $ \bm U^{(q)}$ (at an additional computation cost for on-the-fly computation of $(\bm U^{(1)})^q $), without affecting the asymptotic convergence properties of the method in the large data limit, but introduces a risk of numerical instability at large $q $ (e.g., if $ \bm U^{(1)} $ has eigenvalues with positive real part). 

Next, for each of the $S$ elements $ \Xi_{i,N} $ of the averaging partition for the assimilated observable $h$, we compute an $L\times L $ matrix $ \bm E_i $ representing the spectral projection $E_{\bar h_N,L}(\{ \bar a_{i,N} \} )$. For a given $ \Xi_{i,N} $, this is done by first identifying the timestamps in the training data for which $h$ takes values in this set, viz.
    \begin{displaymath}
        N_i = \{ n \in [ 0, N-1] : h(x_n) \in \Xi_{i,N} \},
    \end{displaymath}
    and then computing the matrix elements
    \begin{displaymath}
        \bm E_{i,jk} = \langle \phi_{j,N}, E_{\bar h_N }(\{ \bar a_{i,N} \}) \phi_{k,N} \rangle_{\mu_N} 
        = \frac{1}{N} \sum_{n \in N_i} \vec \phi_{j,n} \vec \phi_{k,n},
    \end{displaymath}
where $0 \leq i,j \leq L-1$. As with the Koopman matrices $\bm U^{(q)}$, the computational cost of forming the $\bm E_i$ is $O(NL^2)$.

\subsection{Sequential data assimilation}

The necessary ingredients to perform QMDA given discrete-time observations of $h$ are the $L\times L $ matrices $\bm U^{(q)}$ and $\bm E_i$, representing the Koopman operator and spectral projectors, respectively, as well as $L \times L $ matrices $ \bm \rho $ and $ {\bm \rho^+}$, containing the matrix elements of the density operators $ \hat \rho_{t,N}$ and $ \hat \rho^+_{i,N} $ between observations and immediately after observations of $h $, respectively (see Section~\ref{secOpApprox}). In particular, suppose that observations of $h$ are made every $ q \, \Delta t $ time units, with $ q $ a positive integer, and right after a measurement $h(t_n) \in \Xi_{i,N}$ at time $ t_n $, $ n \in \mathbb{N}_0 $, the density matrix is equal to $ \hat{\bm \rho}^+ $, where $ \bm \rho^+_{jk} = \langle \phi_{j,N}, \hat \rho^+_{i,N} \phi_{k,N} \rangle_{\mu_N}$, $0 \leq j,k \leq L - 1$. Then, the density matrix immediately before the measurement at time $ t_{n+1} $ is given by 
\begin{displaymath}
    \bm \rho =\frac{ \bm U^* \bm \rho^+ \bm U }{ \tr( \bm U^* \bm \rho^+ \bm U ) },   
\end{displaymath}
where $ \bm \rho_{jk} = \langle \phi_{j,N}, \hat \rho_{t,N} \phi_{k,N} \rangle_{\mu_N} $ and $ 0 \leq j,k \leq L -1 $. Moreover, the measurement probability for $h $ to lie in interval $ \Xi_{i,N}$ is determined via
\begin{displaymath}
    \hat P_{i,N}(t_{n+1}) =  \tr(\bm E_i \bm \rho).
\end{displaymath}
When the measurement of $ h $ at time $ t_{n+1} $ is made, and found to lie, say, in interval $ \Xi_{i,N} \in \Xi_N$, the density matrix $ \bm \rho $ is updated to obtain a new density matrix $ \bm \rho^+ $, given by
\begin{displaymath}
    \bm \rho^+ = \frac{\bm E_i \bm \rho \bm E_i}{\tr(\bm E_i \bm \rho \bm E_i )}. 
\end{displaymath}
The data assimilation cycle described above is then repeated using the updated density matrix $ \bm \rho^+ $.   

The computational cost to compute $ \bm \rho $ from $ \bm \rho^+ $ by forward evolution with the Koopman operators, and to update $ \bm \rho^+ $ to $ \bm \rho $ by spectral projection is dominated by matrix-matrix multiplication of $L \times L $ matrices, and is thus $O(L^3)$. The cost to compute the measurement probabilities $\hat P_{i,N}(t)$ for all $S$ elements of $\Xi $ is $O(SL) $. As previously stated, a key aspect of this procedure is that following the offline computations of the basis and operator representations in Sections~\ref{appBasis} and~\ref{appOps}, respectively, the computation cost becomes decoupled from the dimension $m$ of the ambient data space and the number $N$ of training samples, depending only on the spectral resolution parameter $L$ and the size of the partition $S$. This is particularly advantageous in real-time applications (e.g., short-term precipitation forecasting), where computational wall-clock time must be significantly smaller than physical time between observations in order for data assimilation to provide useful information.


%
\end{document}